\documentclass{sig-alternate-05-2015}

\usepackage[T1]{fontenc}

\usepackage{amsmath,bm}
\usepackage{amssymb}
\usepackage{amsthm}

\usepackage{xcolor}
\usepackage{xspace}

\usepackage{subcaption}

\usepackage{enumitem}

\usepackage{graphicx}

\usepackage[binary-units=true]{siunitx}

\usepackage{algorithm}
\usepackage{algorithmicx}
\usepackage{algpseudocode}
\usepackage{multicol}
\usepackage{cite}
\usepackage{url}

\usepackage{times}

\newcommand{\unit}[1]{\,\mathrm{#1}}

\newtheorem{thm}{Theorem}
\newtheorem{tdef}{Definition}

\makeatletter
\def\thm@space@setup{%
  \thm@preskip=\parskip \thm@postskip=0pt
}
\makeatother

\DeclareMathOperator{\negl}{negl}

\usepackage{tikz}
\usetikzlibrary{shapes}
\newcommand*\circled[1]{\tikz[baseline=(char.base)]{
  \node[shape=circle,draw,inner sep=0.3pt] (char) {#1};}}

\makeatletter
\def\@copyrightspace{\relax}
\makeatother

\begin{document}

\newcommand{\name}{HardIDX\xspace}

\newcommand{\rnd}{\leftarrow}
\renewcommand{\N}{\mathbb{N}}

\newcommand{\tuple}[1]{\ensuremath{\big(#1\big)}}

\newcommand{\secparSym}{\lambda}
\newcommand{\secpar}{\ensuremath{\lambda}\xspace}

\newcommand{\ie}{i.e.,\ }
\newcommand{\eg}{e.g.,\ }

\newcommand{\TBD}[1]{\textcolor{red}{\newline TO BE DONE:  {#1} \newline}}
\newcommand{\todo}[1]{\textcolor{red}{[#1]}}
\renewcommand{\Pr}[1]{\mathrm{Pr}\left[#1\right]} 
\newcommand{\new}[1]{\textcolor{blue}{#1}}

\newcommand{\valueBSym}{v}
\newcommand{\valueB}{\ensuremath{\valueBSym}\xspace}
\newcommand{\valueBEl}[1]{\ensuremath{\valueBSym_{#1}}\xspace}
\newcommand{\valueBSet}{\ensuremath{\mathbf{\valueBSym}}\xspace}
\newcommand{\valueBSetResult}{\ensuremath{\mathbf{\valueBSym'}}\xspace}

\newcommand{\cipherSym}{C}
\newcommand{\cipher}			{\ensuremath{\cipherSym}\xspace}
\newcommand{\cipherEl}[1]		{\ensuremath{\cipherSym_{#1}}\xspace}
\newcommand{\cipherK}[2]		{\ensuremath{\cipherSym_{#1#2}}\xspace}
\newcommand{\cipherT}			{\ensuremath{\cipherSym_t}\xspace}
\newcommand{\cipherTEl}[1]		{\ensuremath{\cipherSym_t^{#1}}\xspace}
\newcommand{\cipherTSet}		{\ensuremath{\mathbf{\cipherSym_t}}\xspace}
\newcommand{\cipherM}			{\ensuremath{\cipherSym_{\valueB}}\xspace}
\newcommand{\cipherMEl}[1]		{\ensuremath{\cipherSym_{\valueB}^{#1}}\xspace}
\newcommand{\cipherMSet}		{\ensuremath{\mathbf{\cipherSym_{\valueB}}}\xspace}
\newcommand{\cipherSet}			{\ensuremath{\bm{\cipherSym}}\xspace}
\newcommand{\cipherSetRes}		{\ensuremath{\bm{\cipherSym'}}\xspace}

\newcommand{\skeySym}{SK}
\newcommand{\skey}{\ensuremath{\skeySym}\xspace}
\newcommand{\skeyVar}[1]{\ensuremath{\skeySym_{#1}}\xspace}
\newcommand{\skeys}{\ensuremath{SKs}\xspace}
\newcommand{\skeyT}{\ensuremath{\skeySym_{\keyB{}}}\xspace}
\newcommand{\skeyM}{\ensuremath{\skeySym_{\valueB}}\xspace}
\newcommand{\range}{\ensuremath{R}\xspace}
\newcommand{\rangeS}{\ensuremath{R_s}\xspace}
\newcommand{\rangeE}{\ensuremath{R_e}\xspace}
\newcommand{\domain}{\ensuremath{D}\xspace}
\newcommand{\token}{\ensuremath{\tau}\xspace}
\newcommand{\tokenP}{\ensuremath{\tau_{Plain}}\xspace}
\newcommand{\tokenS}{\ensuremath{\tau_s}\xspace}
\newcommand{\tokenE}{\ensuremath{\tau_e}\xspace}
\newcommand{\botSym}{\ensuremath{\bot}\xspace}


\newcommand{\BPtree}{$B^+$-tree\xspace}
\newcommand{\BPtrees}{$B^+$-trees\xspace}
\newcommand{\Btree}{$B$-tree\xspace}
\newcommand{\Btrees}{$B$-trees\xspace}

\newcommand{\maxEl}{\ensuremath{\infty}\xspace}
\newcommand{\minEl}{\ensuremath{-\maxEl}\xspace}
\newcommand{\dummyP}{\ensuremath{dummyPointer}\xspace}

\newcommand{\treeH}{\ensuremath{h}\xspace}
\newcommand{\treeSym}		{T}
\newcommand{\tree}			{\ensuremath{\treeSym}\xspace}
\newcommand{\treeEncSym}	{\gamma}
\newcommand{\treeEnc}		{\ensuremath{\treeEncSym}\xspace}
\newcommand{\treeEncSize}	{\ensuremath{\treeEncSym.{size}}\xspace}

\newcommand{\childrenSym}{ch}
\newcommand{\children}{\ensuremath{\childrenSym}\xspace}
\newcommand{\childrenAmount}{\ensuremath{\#\mathbf{\childrenSym}}\xspace}

\newcommand{\keyBSym}{\ensuremath{k}}
\newcommand{\keyB}[1]{\ensuremath{\keyBSym_{#1}}\xspace}
\newcommand{\keySet}{\ensuremath{\mathbf{\keyBSym}}\xspace}
\newcommand{\keySetResult}{\ensuremath{\mathbf{\keyBSym'}}\xspace}

\newcommand{\inputSet}{\ensuremath{\mathbf{s}}\xspace}
\newcommand{\inputEl}[1]{\ensuremath{s_{#1}}\xspace}
\newcommand{\valueBAmount}{\ensuremath{n}\xspace}
\newcommand{\fileSet}{\ensuremath{\mathbf{f}}\xspace}

\newcommand{\file}[1]{\ensuremath{f_{#1}}\xspace}

\newcommand{\branchingF}{\ensuremath{b}\xspace}

\newcommand{\counterSym}	{ctr}
\newcommand{\counter}		{\ensuremath{\counterSym}\xspace}

\newcommand{\pageSym}		{\omega}
\newcommand{\page}			{\ensuremath{\pageSym}\xspace}
\newcommand{\pageEl}[1]		{\ensuremath{\pageSym_{#1}}\xspace}
\newcommand{\pageAmount}	{\ensuremath{\#\mathbf{\pageSym}}\xspace}
\newcommand{\pageSet}		{\ensuremath{\boldsymbol{\Omega}}\xspace}

\newcommand{\nodeSym}		{x}
\newcommand{\nodeSetSym}	{X}
\newcommand{\nodeSymT}		{y}
\newcommand{\nodeSetSymT}	{Y}
\newcommand{\node}			{\ensuremath{\nodeSym}\xspace}
\newcommand{\nodeT}			{\ensuremath{\nodeSymT}\xspace}
\newcommand{\nodeEl}[1]		{\ensuremath{\nodeSym_{#1}}\xspace}
\newcommand{\nodeElT}[1]	{\ensuremath{\nodeSymT_{#1}}\xspace}
\newcommand{\nodePos}[1]	{\ensuremath{p_{\nodeSym_{#1}}}\xspace}
\newcommand{\nodeSet}		{\ensuremath{\mathbf{\nodeSetSym}}\xspace}
\newcommand{\nodeSetT}		{\ensuremath{\mathbf{\nodeSetSymT}}\xspace}
\newcommand{\nodeSize}		{\ensuremath{\nodeSym.{size}}\xspace}
\newcommand{\nodeAmount}	{\ensuremath{\#\mathbf{\nodeSym}}\xspace}
\newcommand{\nodeAmountT}	{\ensuremath{\#\mathbf{\nodeSymT}}\xspace}
\newcommand{\nodeId}		{\ensuremath{\nodeSym.id}\xspace}
\newcommand{\nodeIdEl}[1]	{\ensuremath{\nodeEl{#1}.id}\xspace}
\newcommand{\nodeIdElT}[1]	{\ensuremath{\nodeElT{#1}.id}\xspace}
\newcommand{\nodeIdElVarT}[2]	{\ensuremath{{#1}_{#2}.id}\xspace}
\newcommand{\nodeIdCh}[1]	{\ensuremath{\nodeSym.chId_{#1}}\xspace}
\newcommand{\nodeMsgHash}[1]{\ensuremath{\nodeSym.hash_{#1}}\xspace}
\newcommand{\nodeNum}[1]	{\ensuremath{\nodeSym_{#1}}\xspace}
\newcommand{\nodeKeyNumber}	{\ensuremath{(\branchingF - 1)}\xspace}
\newcommand{\nodeKeySym}	{k}
\newcommand{\nodeKey}[1]	{\ensuremath{\nodeSym.\nodeKeySym_{#1}}\xspace}
\newcommand{\nodeKeyAmount}	{\ensuremath{\nodeSym.\#\mathbf{\nodeKeySym}}\xspace}
\newcommand{\nodeKeyVar}[2]	{\ensuremath{\nodeSym_{#1}.\nodeKeySym_{#2}}\xspace}
\newcommand{\nodeKeyVarT}[3]{\ensuremath{#1_{#2}.\nodeKeySym_{#3}}\xspace}
\newcommand{\nodePointer}[1]{\ensuremath{\nodeSym.p_{#1}}\xspace}
\newcommand{\nodePointerSet}{\ensuremath{\nodeSym.\bm{p}}\xspace}
\newcommand{\nodeParent}[2]	{\ensuremath{\nodeSym_{#1} {\rightarrow} parent_{#2}}\xspace}
\newcommand{\nodeIsLeaf}	{\ensuremath{\nodeSym.isLeaf}\xspace}
\newcommand{\nodeIsLeafVar}[1]	{\ensuremath{\nodeSym_{#1}.isLeaf}\xspace}
\newcommand{\nodeIsLeafVarT}[2]	{\ensuremath{{#1}_{#2}.isLeaf}\xspace}
\newcommand{\nodeNext}		{\ensuremath{\nodeSym.next}\xspace}
\newcommand{\nodeNextVar}[1]{\ensuremath{\nodeSym_{#1}.next}\xspace}
\newcommand{\nodeNextVarT}[2]{\ensuremath{#1_{#2}.next}\xspace}

\newcommand{\pointerSym}		{p}
\newcommand{\pointer}[1]		{\ensuremath{\pointerSym_{#1}}\xspace}
\newcommand{\pointerSetSym}		{P}
\newcommand{\pointerSet}		{\ensuremath{\mathbf{\pointerSetSym}}\xspace}
\newcommand{\pointerSetEl}[1]	{\ensuremath{\mathbf{\pointerSetSym}_{#1}}\xspace}

\newcommand{\searchKey}{\ensuremath{searchKey}\xspace}
\newcommand{\searchKeyStart}{\ensuremath{searchKeyStart}\xspace}
\newcommand{\searchKeyEnd}{\ensuremath{searchKeyEnd}\xspace}

\newcommand{\HSBTVar}[1]						{\ensuremath{\mathtt{#1}}\xspace}
\newcommand{\HSBTSetupVar}[1]					{\ensuremath{\mathtt{#1\_}\-\mathtt{Setup}}\xspace}
\newcommand{\HSBTSetupParaVar}[2]				{\ensuremath{\mathtt{#1\_}\-\mathtt{Setup}(\allowbreak 1^\secparSym)}\xspace}
\newcommand{\HSBTEncVar}[1]						{\ensuremath{\mathtt{#1\_}\-\mathtt{Enc}}\xspace}
\newcommand{\HSBTEncParaVar}[2]					{\ensuremath{\mathtt{#1\_}\-\mathtt{Enc}(\allowbreak #2)}\xspace}
\newcommand{\HSBTTokVar}[1]						{\ensuremath{\mathtt{#1\_}\-\mathtt{Tok}}\xspace}
\newcommand{\HSBTTokParaVar}[2]					{\ensuremath{\mathtt{#1\_}\-\mathtt{Tok}(\allowbreak #2)}\xspace}
\newcommand{\HSBTDecVar}[1]						{\ensuremath{\mathtt{#1\_}\-\mathtt{Dec}}\xspace}
\newcommand{\HSBTDecParaVar}[2]					{\ensuremath{\mathtt{#1\_}\-\mathtt{Dec}(\allowbreak #2)}\xspace}
\newcommand{\HSBTSearchRangeVar}[1]				{\ensuremath{\mathtt{#1\_}\-\mathtt{Search}\-\mathtt{Range}}\xspace}
\newcommand{\HSBTSearchRangeParaVar}[2]			{\ensuremath{\mathtt{#1\_}\-\mathtt{Search}\-\mathtt{Range(\allowbreak #2)}}\xspace}
\newcommand{\HSBTSearchRangeTrustedVar}[1]		{\ensuremath{\mathtt{#1\_}\-\mathtt{Search}\-\mathtt{Range\_}\-\mathtt{Trusted}}\xspace}
\newcommand{\HSBTSearchRangeTrustedParVar}[2]	{\ensuremath{\mathtt{#1\_}\-\mathtt{Search}\-\mathtt{Range\_}\-\mathtt{Trusted(\allowbreak #2)}}\xspace}

\newcommand{\HSBTSym}						{HSBT}
\newcommand{\HSBT}							{\HSBTVar 						{\HSBTSym}}
\newcommand{\HSBTSetup}						{\HSBTSetupVar					{\HSBTSym}}
\newcommand{\HSBTSetupPara}					{\HSBTSetupParaVar				{\HSBTSym}{1^\secparSym}}
\newcommand{\HSBTEnc}						{\HSBTEncVar					{\HSBTSym}}
\newcommand{\HSBTEncPara}					{\HSBTEncParaVar				{\HSBTSym}{\skey, \inputSet}}
\newcommand{\HSBTTok}						{\HSBTTokVar					{\HSBTSym}}
\newcommand{\HSBTTokPara}					{\HSBTTokParaVar				{\HSBTSym}{\skey, \range}}
\newcommand{\HSBTDec}						{\HSBTDecVar					{\HSBTSym}}
\newcommand{\HSBTDecPara}					{\HSBTDecParaVar				{\HSBTSym}{\skey, \cipherSetRes}}
\newcommand{\HSBTSearchRange}				{\HSBTSearchRangeVar			{\HSBTSym}}
\newcommand{\HSBTSearchRangePara}			{\HSBTSearchRangeParaVar		{\HSBTSym}{\token}}
\newcommand{\HSBTSearchRangeTrusted}		{\HSBTSearchRangeTrustedVar		{\HSBTSym}}
\newcommand{\HSBTSearchRangeTrustedPara}	{\HSBTSearchRangeTrustedParVar	{\HSBTSym}{\token}}

\newcommand{\HSBTOSym}						{\HSBTSym1}
\newcommand{\HSBTO}							{\HSBTVar 						{\HSBTOSym}}
\newcommand{\HSBTOSetup}					{\HSBTSetupVar					{\HSBTOSym}}
\newcommand{\HSBTOSetupPara}				{\HSBTSetupParaVar				{\HSBTOSym}{1^\secparSym}}
\newcommand{\HSBTOEnc}						{\HSBTEncVar					{\HSBTOSym}}
\newcommand{\HSBTOEncPara}					{\HSBTEncParaVar				{\HSBTOSym}{\skey, \inputSet}}
\newcommand{\HSBTOTok}						{\HSBTTokVar					{\HSBTOSym}}
\newcommand{\HSBTOTokPara}					{\HSBTTokParaVar				{\HSBTOSym}{\skeyT, \range}}
\newcommand{\HSBTODec}						{\HSBTDecVar					{\HSBTOSym}}
\newcommand{\HSBTODecPara}					{\HSBTDecParaVar				{\HSBTOSym}{\skeyM, \cipherSetRes}}
\newcommand{\HSBTOSearchRange}				{\HSBTSearchRangeVar			{\HSBTOSym}}
\newcommand{\HSBTOSearchRangePara}			{\HSBTSearchRangeParaVar		{\HSBTOSym}{\token}}
\newcommand{\HSBTOSearchRangeTrusted}		{\HSBTSearchRangeTrustedVar		{\HSBTOSym}}
\newcommand{\HSBTOSearchRangeTrustedPara}	{\HSBTSearchRangeTrustedParVar	{\HSBTOSym}{\token}}

\newcommand{\HSBTTSym}						{\HSBTSym2}
\newcommand{\HSBTT}							{\HSBTVar 						{\HSBTTSym}}
\newcommand{\HSBTTSetup}					{\HSBTSetupVar					{\HSBTTSym}}
\newcommand{\HSBTTSetupPara}				{\HSBTSetupParaVar				{\HSBTTSym}{1^\secparSym}}
\newcommand{\HSBTTEnc}						{\HSBTEncVar					{\HSBTTSym}}
\newcommand{\HSBTTEncPara}					{\HSBTEncParaVar				{\HSBTTSym}{\skey, \inputSet}}
\newcommand{\HSBTTTok}						{\HSBTTokVar					{\HSBTTSym}}
\newcommand{\HSBTTTokPara}					{\HSBTTokParaVar				{\HSBTTSym}{\skeyT, \range}}
\newcommand{\HSBTTDec}						{\HSBTDecVar					{\HSBTTSym}}
\newcommand{\HSBTTDecPara}					{\HSBTDecParaVar				{\HSBTTSym}{\skeyM, \cipherSetRes}}
\newcommand{\HSBTTSearchRange}				{\HSBTSearchRangeVar			{\HSBTTSym}}
\newcommand{\HSBTTSearchRangePara}			{\HSBTSearchRangeParaVar		{\HSBTTSym}{\token, \treeEnc}}
\newcommand{\HSBTTSearchRangeTrusted}		{\HSBTSearchRangeTrustedVar		{\HSBTTSym}}
\newcommand{\HSBTTSearchRangeTrustedPara}	{\HSBTSearchRangeTrustedParVar	{\HSBTTSym}{\token, \nodeSet}}

\newcommand{\ranPermSym}		{\Pi}
\newcommand{\ranPerm}			{\ensuremath{\ranPermSym}\xspace}
\newcommand{\ranPermDef}		{\ensuremath{\ranPermSym: \{0,1\}^\secparSym \times \{0,1\}^{log_2 \nodeAmount} \allowbreak \rightarrow \allowbreak \{0,1\}^{log_2 \nodeAmount}}\xspace}
\newcommand{\ranPermPara}[2]	{\ensuremath{\ranPermSym(#1,#2)}}

\newcommand{\reservedSpace}{\ensuremath{\mathtt{reservedSpace}}\xspace}
\newcommand{\maxAmount}{\ensuremath{\mathtt{maxAmount}}\xspace}
\newcommand{\timePar}[1]{\ensuremath{t_{#1}}\xspace}

\newcommand{\defScheme}{\ensuremath{\mathcal{D}}\xspace}
\newcommand{\simulator}{\ensuremath{\mathcal{S}}\xspace}
\newcommand{\adversary}{\ensuremath{\mathcal{A}}\xspace}
\newcommand{\bitB}{\ensuremath{\texttt{b}}\xspace}

\newcommand{\real}{\ensuremath{\mathbf{Real}_{\adversary}(\secparSym)}\xspace}
\newcommand{\ideal}{\ensuremath{\mathbf{Ideal}_{\adversary, \simulator}(\secparSym)}\xspace}
\newcommand{\leakage}{\mathcal{L}}
\newcommand{\leakESym}{\leakage_{enc}}
\newcommand{\leakHSym}{\leakage_{hw}}
\newcommand{\leakE}{\ensuremath{\leakESym}\xspace}
\newcommand{\leakH}{\ensuremath{\leakHSym}\xspace}
\newcommand{\leakEPara}{\ensuremath{\leakESym(\inputSet)}\xspace}
\newcommand{\leakHPara}{\ensuremath{\leakHSym(\inputSet, \tree, \range, \timePar{})}\xspace}

\newcommand{\treeAccessPatSym}	{\mathcal{X}}
\newcommand{\treeAccessPat}		{\ensuremath{\treeAccessPatSym}\xspace}
\newcommand{\treeAccessPatPara}	{\ensuremath{\treeAccessPatSym(\inputSet, \tree, \range, \timePar{})}\xspace}
\newcommand{\valueBAccessPatSym}	{\Delta}
\newcommand{\valueBAccessPat}		{\ensuremath{\valueBAccessPatSym}\xspace}
\newcommand{\valueBAccessPatPara}	{\ensuremath{\valueBAccessPatSym(\inputSet, \tree, \range, \timePar{})}\xspace}
\newcommand{\pageAccessPatSym}	{\mathcal{P}}
\newcommand{\pageAccessPat}		{\ensuremath{\pageAccessPatSym}\xspace}
\newcommand{\pageAccessPatPara}	{\ensuremath{\pageAccessPatSym(\inputSet, \tree, \range, \timePar{})}\xspace}

\newcommand{\stateDictSym}			{state}
\newcommand{\stateDict}				{\ensuremath{\stateDictSym}\xspace}
\newcommand{\stateDictParaVar}[1]	{\ensuremath{\stateDictSym[#1]}\xspace}

\newcommand{\randomSym}				{r}
\newcommand{\random}				{\ensuremath{\randomSym}\xspace}
\newcommand{\randomEl}[1]			{\ensuremath{\randomSym_{#1}}\xspace}

\newcommand{\SKR}					{\ensuremath{\widetilde{\skey}}\xspace}

\newcommand{\wedgeB}	{\, \wedge \,}
\newcommand{\inB}		{\, {\in} \,}

\newcommand{\PSESym}{PASE}
\newcommand{\PSE}{\ensuremath{\mathtt{\PSESym}}\xspace}
\newcommand{\PSEGen}{\ensuremath{\mathtt{\PSESym\_}\-\mathtt{Gen}}\xspace}
\newcommand{\PSEEnc}{\ensuremath{\mathtt{\PSESym\_}\-\mathtt{Enc}}\xspace}
\newcommand{\PSEDec}{\ensuremath{\mathtt{\PSESym\_}\-\mathtt{Dec}}\xspace}
\newcommand{\PSEGenPara}{\ensuremath{\mathtt{\PSESym\_}\-\mathtt{Gen}(\allowbreak 1^\secparSym)}\xspace}
\newcommand{\PSEEncPara}{\ensuremath{\mathtt{\PSESym\_}\-\mathtt{Enc}(\allowbreak \skey, \valueB)}\xspace}
\newcommand{\PSEDecPara}{\ensuremath{\mathtt{\PSESym\_}\-\mathtt{Dec}(\allowbreak \skey, \cipher)}\xspace}
\newcommand{\PSEEncParaVar}[2]{\ensuremath{\mathtt{\PSESym\_}\-\mathtt{Enc}(\allowbreak #1, #2)}\xspace}
\newcommand{\PSEDecParaVar}[2]{\ensuremath{\mathtt{\PSESym\_}\-\mathtt{Dec}(\allowbreak #1, #2)}\xspace}

	\title{\name: Practical and Secure Index with SGX}

	\numberofauthors{6} %
	\author{
	\alignauthor 
	Benny Fuhry\\
	       \affaddr{SAP Research}\\
	       \email{\footnotesize{benny.fuhry@sap.com}}\\
	\alignauthor
	Raad Bahmani\\
	       \affaddr{Technische Universit\"at Darmstadt}\\
	       \email{\footnotesize{r.bahmani@trust.tu-darmstadt.de}}
	\alignauthor
	Ferdinand Brasser\\
	       \affaddr{Technische Universit\"at Darmstadt}\\
	       \email{\footnotesize{f.brasser@trust.tu-darmstadt.de}}	
	\and  
	\alignauthor 
	Florian Hahn\\
	       \affaddr{SAP Research}\\
	       \email{\footnotesize{florian.hahn@sap.com}}
	\alignauthor
	Florian Kerschbaum\\
	       \affaddr{University of Waterloo}\\
	       \email{\footnotesize{florian.kerschbaum@uwaterloo.ca}}
	\alignauthor
	Ahmad-Reza Sadeghi\\
	       \affaddr{Technische Universit\"at Darmstadt}\\
	       \email{\footnotesize{a.sadeghi@trust.tu-darmstadt.de}}
	}

	\maketitle


\begin{abstract}
Software-based approaches for search over encrypted data are still either challenged by lack of proper, low-leakage encryption or slow performance.
Existing hardware-based approaches do not scale well due to hardware limitations and software designs that are not specifically tailored to the hardware architecture, and are rarely well analyzed for their security (\eg the impact of side channels). 
Additionally, existing hardware-based solutions often have a large code footprint in the trusted environment susceptible to software compromises.
In this paper we present \name: a hardware-based approach, leveraging Intel's SGX, for search over encrypted data. 
It implements only the security critical core, i.e., the search functionality, in the trusted environment and resorts to untrusted software for the remainder.
\name is deployable as a highly performant encrypted database index: it is logarithmic in the size of the index and searches are performed within a few milliseconds rather than seconds.
We formally model and prove the security of our scheme showing that its leakage is equivalent to the best known searchable encryption schemes.
Our implementation has a very small code and memory footprint yet still scales to virtually unlimited search index sizes, i.e., size is limited only by the general -- non-secure -- hardware resources.
\end{abstract}


\section{Introduction}
\label{sec:intro}
	
Outsourcing the storage and processing of sensitive data to untrusted cloud environment is still considered as too risky due to possible data leakage, government intrusion, and legal liability. %
Cryptographic solutions such as Secure Multiparty Computation (MPC) \cite{BenGol88,GolMic87}
and in particular Fully Homomorphic Encryption (FHE) \cite{gentry_fully_2009} offer high degree of protection by allowing arbitrary computation on encrypted data.  
However, MPC and FHE schemes are still impractical for adoption in large distributed systems~\cite{coron_fully_2011, gentry_homomorphic_2012}. 

Moreover, there are a number of useful applications that only require a small set of operations rather than the universal solutions offered by MPC/FHE.
A prime example of such operations is the search and retrieval in databases without the need to download all data from the cloud.
For searching over encrypted data, different cryptographic schemes have been proposed such as property-preserving encryption \cite{bellare_deterministic_2007,boldyreva_ope_2009}, or functional encryption \cite{boneh_functional_2011} and its special case searchable encryption \cite{song_practical_2000,curtmola_searchable_2006,lu_privacy-preserving_2012,KamMoa16}. 
In this context, performing efficient and secure \emph{range} queries are commonly considered to be very challenging.
CryptDB~\cite{popa_cryptdb:_2011} resorts to order-preserving encryption for this purpose which is susceptible to simple ciphertext-only attacks as shown by Naveed et al.~\cite{naveed_inference_2015}.

As we will elaborate in our related work section, many schemes for search over encrypted data supporting range queries require search time linear in the number of database records.
Recently, schemes with polylogarithmic search time, based on an index structure, have been proposed~\cite{demertzis_2016,faber2015rich,lu_privacy-preserving_2012}.
Nonetheless, the first scheme in \cite{lu_privacy-preserving_2012} is not yet practical, because it applies pairing-based cryptography and also leaks sensitive information about the plaintext, namely the order of the plaintexts. 
In contrast, \cite{demertzis_2016} and \cite{faber2015rich} presented approaches with polylogarithmic search time that utilize only lightweight cryptography, that is, pseudorandom functions and symmetric encryption.
Out of the many schemes presented in ~\cite{demertzis_2016}, the most secure approach, without false positives and bearable storage cost achieves practical deployability.
However, it still leaks much sensitive information, \eg the search pattern and the range size of each query. 
Hence, designing an \emph{efficient} searchable encryption scheme with \emph{minimal leakage on the queried ranges} remains an open challenge.

Another line of research~\cite{bajaj_trusteddb:_2014,baumann_shielding_2014} leverages the developments in hardware-assisted Trusted Execution Environments (TEEs) for search over encrypted data. 
Although Intel's recently introduced Software Guard Extension (SGX)~\cite{Intel_SGX3,Intel_SGX2,Intel_SGX1,SGX_Ref,iscaTut,costanintel} has inspired new interest in TEEs, related technologies have been available before, \eg in ARM processors known as ARM TrustZone~\cite{trustzone} as well as in academic research~\cite{SPP10,KSS+14,BEK+15}.
Also, AMD has recently announced a TEE for their CPUs~\cite{amd} rising the hope that TEEs will be widely available in x86 processors, and thus in many relevant environments such as clouds, in the near future.

TEEs provide the secure isolation of sensitive data and computation in hostile environments. 
For instance, SGX loads application code and data from main memory into an isolated memory, which assures confidentiality and integrity of its data against malicious software on the same system, including privileged software like operating system (OS) or hypervisor.
These properties are only guaranteed as long as the TEE operates completely self-contained. 
However, TEEs have to interact with untrustworthy components within the same computer system for various reasons.
In particular, TEEs often have limited resources and need to utilize untrusted resources of their host system. 
For instance, in case of limited memory it swaps out data to untrusted memory or storage.
Hence, in order to achieve comprehensive security, information leakage through those channels has to be considered and taken care of. 
Previous solutions for isolating databases with SGX load and execute the entire unmodified database management system (DBMS) into an enclave~\cite{bajaj_trusteddb:_2014,baumann_shielding_2014}, but they do not consider possible side channels occurring from the use of untrusted resources as well as resource sharing between TEE and untrusted host. 
In particular, previous works do not formally consider information leakage.
Additionally, they do not scale well due to the limited memory size of SGX's enclaves and the large footprint of the code they require in the TEE.

\textbf{Our goal and contribution.} In this paper, we present an efficient scheme for search over encrypted data using SGX that can directly be deployed as a database index.
We utilize SGX's protection characteristics to achieve an outstanding tradeoff between security, performance and functionality.
The currently fastest software-based schemes that support range queries are~\cite{demertzis_2016} and ~\cite{faber2015rich}.
Our solution significantly improves over these approaches in terms of performance and storage.
Compared to the latest hardware-based schemes~\cite{bajaj_trusteddb:_2014,baumann_shielding_2014}, we improve in terms of security and scalability.
Our scheme organizes data in a \BPtree~\cite{bayer_organization_1970} structure that is frequently used for databases indexes in most database management systems (DMBSs)~\cite{_mysql_????,RamakrishnanGehrke200208}, for data access \cite{_power_????} as well as for file systems \cite{giampaolo_practical_1998,270_building_????} to vastly improve the performance of search operations.
Our solution support searches for single values and value ranges.
Additionally, our solution can be adapted easily to many other database (search) operations.

We provide two constructions of our search scheme differing in the management of the \BPtree.
The tree is either loaded and handled as a whole inside the enclave, or is loaded in an on-demand fashion (see Section~\ref{sec:design}).
Although the tree does not directly leak information about the content of the (encrypted) index, it can leak access patterns information observed during query processing.
For both constructions we show the access pattern leakage with respect to side channels observable by an attacker.
Furthermore, we show that the leakage of both construction is almost the same differing only by their granularity.

We implemented and extensively evaluated both constructions on SGX-enabled hardware (see Section~\ref{sec:impl} and Section~\ref{sec:eval}). 
The first construction loads the whole \BPtree structure in an SGX enclave and performs search queries thereafter.
However, it is not possible to load arbitrarily large structures into an enclave, as enclave size is limited.
We will detail this in Section~\ref{sec:impl}.
Our second construction only loads those data in SGX that is currently needed to process a search query.
It has a very small code and memory footprint in the TEE compared to other hardware-based approaches \cite{bajaj_trusteddb:_2014,baumann_shielding_2014}.
Additionally, it scales to arbitrary index sizes, as memory usage in the enclave is constant and untrusted resources are used to store the database itself.

\noindent 
Our main contributions are as follows: 
\begin{itemize}[leftmargin=0.4cm, itemsep=-2pt,topsep=0pt]	
	\item Our scheme has logarithmic complexity in the size of index and searches are performed within a few milliseconds.
	\item We formally model and prove our scheme secure showing that its security (leakage) is comparable to the best known searchable encryption schemes.
	\item We provide an implementation and evaluate the performance and functional bottleneck of SGX on the basis of two different constructions that are designed specifically for SGX to reduce the Trusted Computing Base.
\end{itemize}
	

\section{Background}
\label{sec:background}

\begin{figure*}[h]
\centering
\begin{subfigure}{.64\textwidth}
  \centering
  \includegraphics[width=\textwidth]{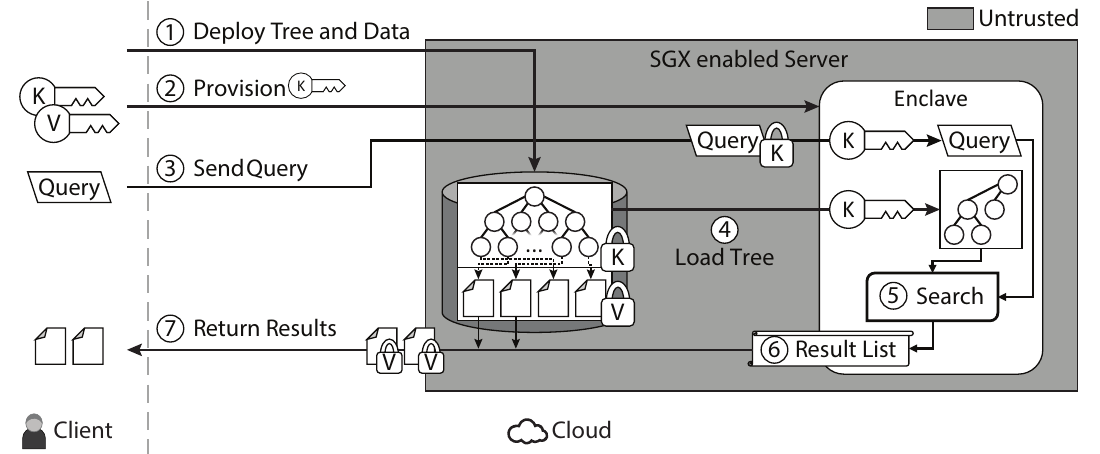}
  \caption{High level design}
  \label{fig:overview}
\end{subfigure}\hfil%
\begin{subfigure}{.3\textwidth}
  \centering
  \includegraphics[width=\textwidth]{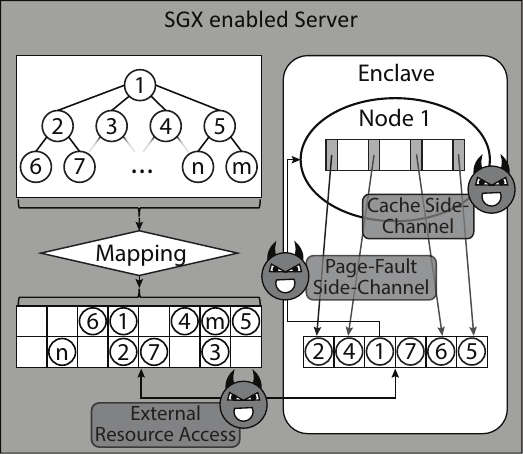}
  \caption{Attack channels}
  \label{fig:attacker}
\end{subfigure}
\caption{System model}
\label{fig:model}
\end{figure*}

We will start by briefly introducing SGX towards which we design \name.
Note, our solution is not strictly limited to SGX and could be used with any other system having a TEE with capabilities similar to SGX (e.g., Sanctum~\cite{sanctum}). 
Afterwards we explain the basics of side channel attacks, which are a major concern for \name.

\subsection{Intel Software Guard Extensions (SGX)}
\label{sec:overview:subsec:sgx}
SGX is an extension of the x86 instruction set architecture (ISA) 
introduced with the 6th Generation Intel Core processors (code name Skylake).
We now present a high level overview of SGX's features utilized by \name (see \cite{Intel_SGX3,Intel_SGX2,Intel_SGX1,SGX_Ref,iscaTut,costanintel} for more details).

\textbf{Memory Isolation.} 
On SGX enabled platforms, programs can be divided into two parts, an \emph{untrusted part} and an isolated, \emph{trusted part}.
The trusted part, called \emph{enclave} in SGX terminology, is located in a dedicated portion of the physical RAM. 
The SGX hardware enforces additional protection on this part of the memory.
In particular, all other software on the system, including privileged software like OS, hypervisor, firmware and code in system management mode (SMM) cannot access the enclave memory. 

The untrusted part is executed as an ordinary process within the virtual memory address space and the enclave memory is mapped into the virtual memory of the untrusted \emph{host process}.
This mapping allows the enclave to access the entire virtual memory of its host process, while the (untrusted) host process can invoke the enclave only through a well-defined interface.
Furthermore, all isolated code and data is encrypted while residing outside of the CPU package.
Decryption and integrity checks are performed when the data is loaded inside the CPU.

\textbf{Memory Management.} 
SGX dedicates a fixed amount of the system's main memory (RAM) for enclaves and related metadata.
For current systems this memory is limited to \SI{128}{\mega\byte} which is used for both, SGX metadata and the memory for the enclaves themselves. 
The latter is called Enclave Page cache (EPC) and is about \SI{96}{\mega\byte} according to our evaluations.
The SGX memory is reserved in the early boot phase and is static throughout the runtime of the system.
As the number of enclaves which may be loaded and executed in parallel is virtually unlimited, the OS has to manage the enclave memory dynamically.
The OS can allocate (parts of) the memory to individual enclaves and change these allocation during the runtime of the enclaves. 
In particular, the OS can swap out enclave pages. 
SGX ensures integrity, confidentiality and freshness of swapped-out pages.

\textbf{Attestation.} 
SGX has a remote attestation feature which allows to verify the correct creation of an enclave on a remote system.
During enclave creation the initial code and data loaded into the enclave are measured.
This measurement can be provided to an external party to prove the correct creation of an enclave.
The authenticity of the measurement as well as the fact that the measurement originates from a benign enclave is ensured by a signature, provided by SGX's attestation feature (refer to~\cite{Intel_SGX3} for details). 
This signature is provided by a component of SGX, called quoting enclave (QE).
The QE accepts only measurements from the hardware and the hardware ensures that only correct enclaves are measured.
Furthermore, the remote attestation feature allows for establishing a secure channel between an external party and an enclave.

\subsection{Side Channel Attacks}
\label{sec:background:subsec:sidechannel}

Side channel attacks allow an adversary to extract sensitive information without having direct access to the information source by observing effects of the processing of the sensitive information. 
Side channel attacks have been known for a long time and various variants have been studied in the past, including hardware side channels like ground potential, EM or power consumption~\cite{Genkin2015}; software timing side channels that can be observed remotely~\cite{Brumley2011} as well as locally, \eg through cache timing side channels~\cite{YF14,LYG+15,OST06}.
However, all these side channel attacks are noisy and require repeated execution and measurements to extract the sensitive information.

In the context of SGX there exist a new class of side channels, called deterministic side channel~\cite{xu_controlled-channel_2015}.
As the OS is untrusted, yet still manages the enclave's resources, including enclave memory, it can observe the enclaves behavior.
In particular, the OS can generate a precise trace of the enclave's code and data accesses at the granularity of pages.
In~\cite{xu_controlled-channel_2015} it is shown that this allows to extract sensitive information from an SGX enclave.

\section{High Level Design}
\label{sec:highLevelDesign}

Now, we will give a high level description of our design and describe the general working of our scheme. 
Afterwards we explain our attacker model.

\subsection{\name Overview}
\label{sec:highLevelDesign:subsec:overview}

The high level design of our solution is shown in Figure~\ref{fig:overview}. 
The design involves three entities: the client (who is the data owner and therefore trusted), the untrusted SGX enabled server and the trusted SGX enclave \emph{within} the server.

Initially, a client prepares its data values (abbreviate with values in the following) by augmenting it with (index) search keys.
The values are stored at pseudo-random position.
The keys are then inserted into a \BPtree and the storage order of all nodes is also pseudo-random.
The tree and values are linked by adding pointers to the leaves of the tree identifying the random position of the corresponding values.
A value can be any data such as records in a relational database or files/documents in other database types.
The client then encrypts all nodes of the tree with a secret key \skeyT and all data values with \skeyM.
Then, it deploys the encrypted \BPtree and encrypted values on the untrusted server in the cloud (see step~\circled{1} in Figure~\ref{fig:overview}\footnote{For visualization purposes, the tree nodes and values are shown to be encrypted as a block. 
In reality each node and value is encrypted individually.}).

In the second step, a secure connection is established between the client and the enclave.
The client uses the SGX attestation feature for authenticating the enclave (see Section~\ref{sec:impl} for details).
Through this secure connection, the client provisions \skeyT into the enclave (see step~\circled{2}).
This step completes the setup of our scheme, which needs to be executed only once.
Even when the enclave is unloaded, e.g., due to a reboot, the state (including \skeyT) can be securely stored and restored from local memory~\cite{Intel_SGX3}.

From now on the client can send (index) search queries to the server.
Randomized encryption with the key \skeyT is used for all search queries.
Hence, the untrusted server cannot learn anything about the query, not even if the same query was send before.
When a query arrives in the enclave, \skeyT is used to decrypt the query (see step~\circled{3}).

In step~\circled{4}, the enclave loads the \BPtree structure (tree nodes, but no values) from the untrusted storage into enclave memory and decrypts it.
Given sufficient memory, the entire tree is loaded into the enclave and the search is performed afterwards (see step~\circled{5}).

However, as the tree size can exceed the memory available inside the enclave we provide a second design.
In this case, only a subset of tree nodes is loaded into the enclave.
The tree is traversed starting from the root node. 
When the search reaches an edge to a node, which is currently not present in the enclave, it is fetched from the untrusted storage. 

In both cases the search algorithm eventually reaches a set of leaf nodes, which holds pointers to data values matching the query.
This list of pointers, representing the search result, is passed to the untrusted part (see step~\circled{6}).
The untrusted part learns nothing, except for the cardinality of the result set, from this interaction, because the values are stored in a randomized order.

The result of the index search can be processed further, e.g. in combination with additional SQL operators, in the SGX enclave at the server.
In order to complete the end-to-end secure search we assume that the server uses the pointers, in step~\circled{7}, to fetch the encrypted values from untrusted storage and sends them to the client.
The client uses \skeyM to decrypt the received files.

Notably, the plaintext data values are never available on the server.
They are encrypted with strong standard cryptography methods\footnote{Our implementation uses AES-128 in GCM mode} and never decrypted on the server, not even inside the SGX enclave. 
The key to decrypt the data values is only know to the client.

\subsection{Assumptions and Attacker Model}
\label{sec:highLevelDesign:subsec:attackModel}

We assume a system which provides SGX (or any TEE with capabilities similar to SGX).
In particular, code and data inside the TEE are protected with respect to their integrity and confidentiality. 
We further assume that the TEE provides means to establish a secure channel to the client which allows secure communication and secure provisioning.

Our attacker model is illustrated in Figure~\ref{fig:attacker}.
Due to SGX's protection, the attacker cannot directly access the enclave.
However, side channels exist through, which the attacker could potentially extract sensitive information.
In particular, the attacker aims to learn the structure of the \BPtree which represents the order relation between the indexed data.

We assume the attacker has full control over all software on the system running \name.
(1) The attacker can observe all interaction of the enclave with resources outside the enclave.
In particular, the attacker can observe the access pattern to \BPtree nodes stored outside the enclave.
(2) The attacker can use deterministic page-fault side channel to observe data access inside the enclave at page granularity~\cite{xu_controlled-channel_2015}.
Through this side channel, the attacker can observe access patterns on the \BPtree stored \emph{inside} the enclave.
(3) The attacker can use cache side channel to learn about code paths or data access patterns inside the enclave, as SGX does not protect against them~\cite{costanintel}.

Hardware attacks (in particular hardware / physical side channels) are out of scope in this paper.
Furthermore, we consider denial of service (DoS) attacks out of scope, as the untrusted server could always refuse to serve the enclave. 
Similarly, the network connection between the client and the server is vulnerable to DoS attacks.


\section{Notation and definitions}
\label{sec:notation}

	\subsection{\BPtree}

	A \BPtree is a balanced, n-ary search tree.
	So called search keys are utilized to index values.	
	For instance, unique staff ids are used to find the corresponding database record (see Figure~\ref{fig:bPlusTree}).
	An index could also be built for other information, \eg the salary, and a \BPtree is able to retrieve all records below or above a limit or between limits.
	We abbreviate the search keys as keys throughout this paper and clearly differentiate all other key types (\eg cryptographic keys) if ambiguous.

	\begin{figure}[hbt]
		\centering
		\includegraphics[width=\linewidth]{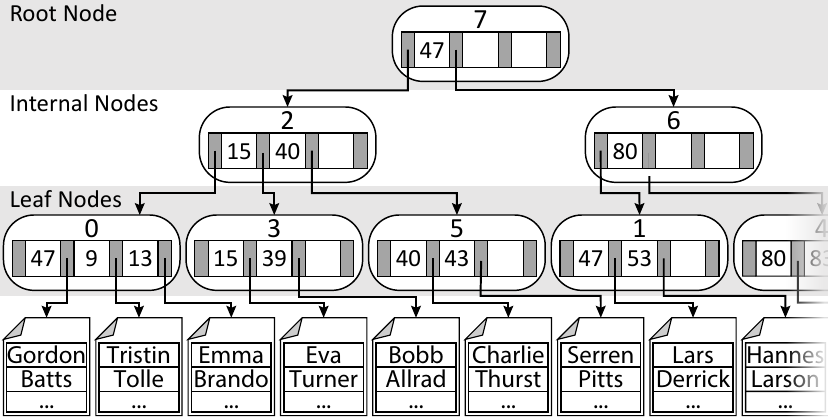}
		\vspace*{-1.6em}
		\caption{\BPtree example (random storage position on the top): the unique staff ids (inside the rectangles) are used as keys and the values are the staff records. The records contain surnames, last names and more attributes.}

		\label{fig:bPlusTree}
	\end{figure}

	Three node types are differentiated in a \BPtree : the root node, internal nodes and leaf nodes.
	Every node \node contains \nodeKeyAmount keys that are stored in a nondecreasing order: $\nodeKey{1} \leq ... \leq \nodeKey{\nodeKeyAmount}$.
	At every inner node \node (including the root if not the only node), the keys separate the key domain into $(\nodeKeyAmount + 1)$ subtrees that are reachable by $(\nodeKeyAmount + 1)$ child pointers: $\{\nodePointer{0}, \allowbreak ..., \allowbreak \nodePointer{(\nodeKeyAmount)}\} = \nodePointerSet$.
	Every key \nodeKey{i} has a corresponding pointer \nodePointer{i} that points to a node containing elements greater than or equal to \nodeKey{i} and smaller than any other tag $\nodeKey{j}$ $\forall j \in [i + 1,\nodeKeyAmount]$.
	\nodePointer{0} points to a node containing only keys, which are smaller than \nodeKey{1}.
	No internal node is linked to a value.
	Instead, every leaf node \node stores \nodeKeyAmount keys and a pointer to its corresponding value (\nodePointer{0} is not used at the leaves).
	We denote the \BPtree without the values as \emph{\BPtree structure}.
	Every node \node in the tree has a unique id \nodeId{}, a flag stored in \nodeIsLeaf if it is a leaf and we denote \nodePos{i} as the storage position of \nodeEl{i}, \ie the physical memory address.

	We use unchained $B^+$-trees, \ie the leafs are not connected.
	Linked leaves would increase the search performance, but it would severely deteriorate the security.
	The reason is that a range query would directly leak the relationship among leaves if links are followed during a query.

	In our constructions, it is not necessary to define the key domain in advance as in many other approaches.
	It is not even necessary that the domain is a range of integers.
	Instead, \domain can be an arbitrary domain with a defined order relation and a defined minimal and a maximal element that can be recognized by the algorithms.
	These two elements, denoted as \minEl and \maxEl, fulfill the following condition: $\minEl < \nodeKey{} < \maxEl$ $\forall \nodeKey{} \in \domain$.	

	The branching factor \branchingF specifies a \BPtree by defining the maximal number of pointers.
	\branchingF also defines the minimal number of pointer for the different node types, but we do not further elaborate on details.
	For ease of exposition, we assume that every key and pointer fits in an \SI{32}{\bit} block, but this is no prerequisite for our constructions.

	\subsection{Probabilistic symmetric encryption:}
	A probabilistic authenticated symmetric encryption consists of three probabilistic poly\-nomial-time algorithms $\PSE=\tuple{\PSEGenPara, \allowbreak  \PSEEncPara, \allowbreak \PSEDecPara}$. 
	It provides confidentiality, integrity and authenticity.
	\PSEGen\xspace takes a security parameter $1^\secpar$ as input and generates a secret key \skey. 
	\PSEEnc takes the key \skey and a value \valueB as input and returns a ciphertext \cipher. 
	\PSEDec takes the key \skey and a ciphertext \cipher as input and returns \keyB{} iff \keyB{} was encrypted with \PSEEnc under the key \skey.
	Otherwise, it returns \botSym. 
	\PSE has to be an authenticated IND-CCA secure encryption, \eg AES-128 in GCM mode.

	\subsection{Hardware Secured \BPtree (\HSBT)}	

	Based on the presented definition of a \BPtree, we define the notion of a Hardware Secured \BPtree as follows.
	We assume that the \BPtree should store a set \inputSet of \valueBAmount key-value pairs: $\inputSet = ((\keyB{1}, \valueBEl{1}), \allowbreak ..., \allowbreak (\keyB{n}, \valueBEl{n}))$. This set consists of \valueBAmount values $\valueBSet = (\valueBEl{1}, ..., \valueBEl{n})$ and their corresponding keys $\keySet = (\keyB{1}, ... \keyB{n})$. 

	\begin{tdef}[\HSBT]\label{thm:shbt} 
		A secure hardware $B^+$-tree scheme is a tuple of six polynomial-time algorithms \tuple{\HSBTSetup, \allowbreak \HSBTEnc, \allowbreak \HSBTTok, \allowbreak \HSBTDec, \allowbreak \HSBTSearchRange, \allowbreak \HSBTSearchRangeTrusted}.
		\vspace*{0.3em}

		\noindent Algorithms executed at the client:
		\begin{description}[style=unboxed,leftmargin=.3cm, itemsep=-2pt,topsep=0pt]
			\item[$\skey \rnd \HSBTSetupPara$:] 
				Take the security parameter \secpar as input and output a secret key \skey.
			\item[$\treeEnc \rnd \HSBTEncPara$:] 
				Take the secret key \skey and a set \inputSet of key-value pairs as input. Output an encrypted $B^+$-tree \treeEnc.
			\item[$\token \rnd \HSBTTokPara$:] 
				Take the secret key \skey and a range $\range = [\rangeS, \rangeE]$ as input. 
				Output a search token \token.	
			\item[$\valueBSetResult \rnd \HSBTDecPara$:] 
				Take the secret key \skey and a set of ciphertext \cipherSetRes as input. 
				Decrypt the ciphertexts and outputs plaintext values \valueBSetResult.	
		\end{description}
		\vspace*{0.3em}

		\noindent Executed at the server on untrusted hardware:
		\begin{description}[style=unboxed,leftmargin=.3cm, itemsep=-2pt,topsep=0pt]
			\item[$\cipherSetRes \rnd \HSBTSearchRangePara$:] 
				Take a search token \token as input and call the secure hardware function \HSBTSearchRangeTrusted.
				Output a set of encrypted values \cipherSetRes.
		\end{description}
		\vspace*{0.3em}

		\noindent Executed at the server on secure hardware:
		\begin{description}[style=unboxed,leftmargin=.3cm, itemsep=-2pt,topsep=0pt]
			\item[$\pointerSet \rnd \HSBTSearchRangeTrustedPara$:] 
				Take a search token \token as input.
				Output a set of pointers \pointerSet.
		\end{description}
	\end{tdef}

	In the following definitions and constructions, we assume a passive attacker.
	This is important, as especially the correctness can easily be thwarted by an active attacker.
	For instance, he can drop results before they are transferred to the client.
	We present implications of an active attacker and countermeasures in Section~\ref{sec:design:subsec:malAtt}.

	\vspace{0.5em}
	\begin{tdef}[Correctness]\label{thm:cor} Let \defScheme denote a \HSBT-scheme consisting of the six algorithms described in Definition \ref{thm:shbt}.
	Given a passive attacker, we say that \defScheme is correct if for all $\secpar \in \N$, for all \skey output by \HSBTSetupPara, for all key-value pairs \inputSet used by \HSBTEncPara to output \treeEnc, for all \range used by \HSBTTokPara to output \token, for all \cipherSetRes output by \HSBTSearchRangePara, the values \valueBSetResult output by \HSBTDecPara are all values in \inputSet for which the corresponding keys \keySetResult fall in \range, \ie $\valueBSetResult = \{\valueBEl{i}|(\keyB{i}, \valueBEl{i}) \in \inputSet \land \keyB{i} \in [\rangeS, \rangeE] = \range\}$.
	\end{tdef}

	Our security model, which we define next, is based on a three step framework introduced by Curtmola et al. in \cite{curtmola_searchable_2006}.
	The first step is to formulate a leakage, \ie an upper bound of the information that an adversary can gather from the protocol.
	Secondly, one defines the \real and a \ideal game for an adaptive adversary \adversary and a polynomial time simulator \simulator.
	\real is the execution of the actual protocol and \ideal utilizes \simulator to simulate the real game by using only the formulated leakage.
	An adaptive adversary can use information learned in previous protocol iterations for its queries.
	Third, a scheme is CKA2-secure if one can show that \adversary can distinguish the output of the games with probability negligibly close to 0.
	This in turn means that \adversary does not learn anything besides the leakage stated in the first step, because otherwise he could use this additional information to distinguish the games.

	At security models of searchable encryption schemes so far, the leakage only covers the transaction between the client and server.
	In our scenario, there is an additional transaction between the server and the secure hardware that can be viewed by the adversary.
	Therefore, we extend the CKA2-security to CKA2-HW-security by introducing a new type of leakage denoted as \leakH.
	It consists of the inherent leakage of the used secure hardware and the inputs/outputs to/from the secure hardware.
	\vspace{0.5em}
	\begin{tdef}[CKA2-HW-security]\label{thm:sec} 
		Let \defScheme denote a \HSBT-scheme consisting of the six algorithms described in Definition \ref{thm:shbt}.
		Consider the probabilistic experiments \real and \ideal, whereas \adversary is a stateful passive adversary and \simulator is a stateful simulator that gets the leakage functions \leakE and \leakH.
		\vspace*{0.3em}

		\begin{description}[style=unboxed,leftmargin=.3cm, itemsep=-2pt,topsep=0pt]
			\item[\real:] 
				the challenger runs \HSBTSetupPara to generate a secret key \skey.
				\adversary outputs a set of key-value pairs $\inputSet = ((\keyB{1}, \valueBEl{1}), ... (\keyB{n}, \valueBEl{n}))$.
				The challenger calculates $\treeEnc \rnd \HSBTEncPara$ and passes \treeEnc to \adversary.
				Afterwards, \adversary makes a polynomial number of adaptive queries for arbitrary ranges \range.
				The challenger returns search tokens \token to \adversary after calculating $\token \rnd \HSBTTokPara$.
				\adversary can use \treeEnc and the returned tokens at any time to make a query to the secure hardware.
				The secure hardware returns a set of pointers \pointerSet.
				Finally, \adversary returns a bit \bitB that is the output of the experiment.
			\item[\ideal:] 
				the adversary \adversary outputs a set of key-value pairs $\inputSet = ((\keyB{1}, \valueBEl{1}), \allowbreak ..., \allowbreak  (\keyB{n}, \valueBEl{n}))$.
				Using \leakE, \simulator creates \treeEnc and passes it to \adversary.
				Afterwards, \adversary makes a polynomial number of adaptive queries for arbitrary ranges \range.
				The simulator \simulator creates tokens \token and passes them to \adversary.
				The adversary \adversary can use \treeEnc and the returned tokens at any time to make queries to \simulator (that simulates the secure hardware).
				\simulator is given \leakH and returns a set of pointers \pointerSet{}.
				Finally, \adversary returns a bit \bitB that is the output of the experiment.
		\end{description}
		\vspace*{0.3em}

		\noindent
		We say \defScheme is \tuple{\leakE,\leakH}-secure against adaptive chosen-keyword attacks if for all probabilistic polynomial-time algorithms \adversary, there exists a probabilistic polynomial-time simulator \simulator such that
		\vspace*{-0.3em}
		\begin{displaymath}
			|\Pr{\real = 1} - \Pr{\ideal = 1}| \leq \negl(\secpar)
		\end{displaymath}
	
	\end{tdef}	

\section{Search Algorithms}
\label{sec:design}

	In this section, we will present two different constructions that enable a client to the search for a single value or a range of values based on keys.
	We use \BPtrees in both constructions to achieve logarithmic search and SGX to protect the confidentiality and integrity of the data.


\subsection{Construction~1}
\label{sec:design:subsec:const1}

	We describe our first correct (according to Definition \ref{thm:cor}) and secure (according to Definition \ref{thm:sec}) construction in this section.
	The guiding idea of the construction is that the entire data should be stored and processed inside the enclave.	
	The client constructs the \BPtree locally, encrypts the \BPtree structure and the values with \skeyT and \skeyM, respectively, and sends both to the cloud provider.
	The SGX application gets deployed to an SGX capable server at a cloud provider (see Section~\ref{sec:overview:subsec:sgx} for details).

	Software measurement is used for remote attestation, \ie to prove to the client that the correct software is deployed on an SGX enabled CPU.
	During the deployment, the application reserves an SGX protected memory region.	
	A secure transfer protocol between client and server is used to deploy \skeyT inside the enclave.
	Thus \skeyT is only known by the enclave's process and the client. 
	The cloud provider and all other processes cannot access this key at any point in time.
	
	The next step is to load the \BPtree structure from an untrusted to the isolated memory region and use \skeyT to decrypt the tree nodes.
	It is important to note that the tree is still protected, because all data inside the enclave is secured by SGX.
	The enclave is then ready to receive search queries from the client.

	Since the values can be very large, the data transfer to the enclave can cause a severe performance overhead.
	We, therefore, store the values outside of the enclave.
	A second reason is the limited size of protected memory inside an enclave.
	Storing the values outside of the enclave has no security implications as the values are encrypted with an authenticated IND-CCA secure encryption scheme and they are stored in a randomized order.
	The untrusted part receives only pointers to the values from the trusted part and loads them from memory or disk by itself.

	We now describe the \HSBT-scheme \HSBTO consisting of the six algorithms \tuple{\HSBTOSetup,\allowbreak \HSBTOEnc,\allowbreak \HSBTOTok,\allowbreak \HSBTODec,\allowbreak \HSBTOSearchRange,\allowbreak \HSBTOSearchRangeTrusted} and utilizing a pseudo\-random permutation \ranPermDef in more detail:

	\begin{description}[style=unboxed, itemsep=-2pt,leftmargin=.3cm,topsep=0pt]			
		\item[$\skey \rnd \HSBTOSetupPara$:] 			
			Use input $\secpar$ to execute \PSEGen two times and output $\skey = \tuple{\skeyT, \skeyM}$ that the two instances of \PSEGenPara output. 
			\skeyM and \skeyT are kept secret at the client.
			\skeyT is additionally shared with the server enclave using a secure transport and deployment protocol.
			The enclave stores \skeyT inside the isolated enclave.
		\item[$\treeEnc \rnd \HSBTOEncPara$:] 
			Take \skey and $\inputSet = ((\keyB{1}, \valueBEl{1}), \allowbreak ..., \allowbreak (\keyB{n},\allowbreak \valueBEl{n}))$ as input.
			Start by storing all values $\valueBSet = (\valueBEl{1}, ..., \valueBEl{n})$ in a random order.
			An almost standard \BPtree insertion is used for all keys.
			One difference is that every newly created node \node gets an id according to the creation order, \ie the first node gets id $0$ ($\nodeIdEl{} = 0$), the second id $1$ ($\nodeIdEl{} = 1$) et cetera.
			After each pair is inserted, the empty position for keys and pointers in the tree get filled up. 
			More specifically, a node \node that contains \nodeKeyAmount keys from the domain gets filled with $(\branchingF - 1 - \nodeKeyAmount)$ keys \maxEl and $(\branchingF - \nodeKeyAmount)$ dummy pointers.	
			Then, all keys and pointers are padded to a length of \SI{32}{\bit} (as mentioned before, this is no prerequisite of our solution).
			The ids are used by the algorithm to store the nodes at pseudorandom positions: $\nodePos{} = \ranPermPara{\skeyT}{\nodeIdEl{}}$.	
			Now, we have a \BPtree in which every node occupies the same storage space and the order of the nodes and values is random.
			Finally, \PSEEncParaVar{\skeyM}{\cdot} is used to encrypt every value and \PSEEncParaVar{\skeyT}{\cdot} is used to encrypt every node. 
			The encrypted nodes and values form the encrypted tree \treeEnc, which is protected by an authenticated IND-CPA secure encryption.		
		\item[$\token \rnd \HSBTOTokPara$:] 
			Use input \skeyT and $\range = [\rangeS, \rangeE]$ to calculate $\token  \rnd \PSEEncParaVar{\skeyT}{\rangeS || \rangeE}$ and output \token.
			Query for all elements below \rangeE or all elements above \rangeS can be created by using $\rangeS = \minEl$ or $\rangeE = \maxEl$, respectively.
		\item[$\valueBSetResult \rnd \HSBTODecPara$:] 
			Use input \skeyM to decrypt the encrypted values $\cipherSetRes = \tuple{\cipherEl{0}, ..., \cipherEl{j}}$: $\valueBSetResult = \tuple{\PSEDecParaVar{\skeyM}{\cipherEl{0}}, \allowbreak ...,  \PSEDecParaVar{\skeyM}{\cipherEl{j}}}$.
			Output \valueBSetResult.
	\end{description}
	\vspace*{0.3em}

	\noindent Executed at the server on untrusted hardware:
	\begin{description}[style=unboxed,leftmargin=.3cm, topsep=0pt]
		\item[$\cipherSetRes \rnd \HSBTOSearchRangePara$:] 
			Take the search token \token as input, pass \token to the trusted part and return the result values (by dereferencing the pointers): $\cipherSetRes = \tuple{\cipherEl{0}, ..., \cipherEl{j}}$.
			See Algorithm~\ref{algo:SHBTSearch1} for details.
	\end{description}
	\vspace*{0.3em}

	\noindent Executed at the server on secure hardware:
	\begin{description}[style=unboxed,leftmargin=.3cm, itemsep=-2pt,topsep=0pt]
		\item[$\pointerSet \rnd \HSBTOSearchRangeTrustedPara$:] 
			Take the search token \token as input, decrypt the token, perform a \BPtree traversal and return the pointers that lead to values falling in the queried range in a random order.
			The list with the currently processed nodes (\nodeSet) is shuffled to further hide the order.
			Note that the enclave does not have to decrypt the nodes, because the client's encryption was removed already and the CPU has plaintext access to them.
			See Algorithm~\ref{algo:SHBTSearchTrusted1} for details.
	\end{description}

	\begin{algorithm}		
		\fontsize{8}{9.5}\selectfont%
		\begin{algorithmic}[1]
			\State $results$ = \Call{HSBT\_SearchRangeTrusted}{\token}  
			\State \textbf{return} *$results_0$, *$results_1$, ...
		\end{algorithmic}
		\caption{\HSBTOSearchRangePara \label{algo:SHBTSearch1}}
	\end{algorithm}

	\begin{algorithm}[ht]		
		\fontsize{8}{9.5}\selectfont%
		\begin{algorithmic}[1]
			\State \tokenP = \PSEDecParaVar{\skeyT}{\token}
			\State parse \tokenP as $\tuple{\rangeS, \rangeE}$

			\State \pointerSet = $\emptyset$ \Comment{Pointer list}  
			\State \nodeSet = $\{root\}$ \Comment{Nodes list}  

			\While{\nodeSet $\neq$ $\emptyset$}
				\State \pointerSetEl{tmp} = $\emptyset$ 
				\State \node = \nodeSet.\Call{pop}{\,}					
			
				\If{not \node.\textit{isLeaf} $\&\&$ \rangeS $<$ \nodeKey{1}}
					\State \pointerSetEl{tmp}.\Call{add}{\nodePointer{0}}
				\EndIf

				\For{$i = 1$; $i < \branchingF - 1$; $i$++}
					\If{
						(\nodeKey{i} $\leq$ \rangeS $<$ \nodeKey{i + 1}) $||$ 
						(\nodeKey{i} $\leq$ \rangeE $<$ \nodeKey{i + 1}) $||$ \newline 
						\phantom a \phantom a \phantom a \phantom a \phantom . (\rangeS $\le$ \nodeKey{i}  $\&\&$ \nodeKey{i + 1} $\le$ \rangeE)   
					}
						\State \pointerSetEl{tmp}.\Call{add}{\nodePointer{i}}
					\EndIf		
				\EndFor

				\If{\nodeKey{b - 1} $\leq$ \rangeE}
					\State \pointerSetEl{tmp}.\Call{add}{\nodePointer{b - 1}}
				\EndIf

				\For{\pointer{} in \pointerSetEl{tmp}}
					\If{\text{*\pointer{}}.\textit{isLeaf}}
						\State \pointerSet.\Call{add}{\pointer{}}
					\Else
						\State \nodeSet.\Call{add}{*\pointer{}}
					\EndIf
				\EndFor
				\State \nodeSet = random permutation of \nodeSet
			\EndWhile			

			\State \pointerSet = random permutation of \pointerSet

			\State \textbf{return} \pointerSet
		\end{algorithmic}
		\caption{\HSBTOSearchRangeTrustedPara \label{algo:SHBTSearchTrusted1}}
	\end{algorithm}

	The construction is correct according to Definition \ref{thm:cor}, because it performs a textbook \BPtree traversal (with an additional randomize step) inside SGX.
	Furthermore, it is based on a correct \PSE-scheme.

	A substantial advantage of the approach to put the whole functionality in an enclave is that SGX provides a high level of protection for the enclave memory and thus the \BPtree.
	One would even assume that an attacker does not learn anything about the tree traversal, because all operations are executed only in protected memory.
	Unfortunately, SGX is not able to fully protect the access pattern.
	
	As mentioned before, we assume that an attacker is able to reveal the accessed pages during the \BPtree traversal.
	Each page allocates \SI{4}{\kilo\byte} and every encrypted node consists of one or multiple AES-blocks.
	Up to $k = \SI{4}{\kilo\byte} / (o \cdot \SI{128}{\bit})$ nodes are contained in one page if $o$ AES-blocks are used by each node.
	Experiments showed that $102$ AES-blocks are used for each node if $\branchingF = 100$ and \SI{32}{\bit} keys and pointers are used.
	Therefore, even multiple of those huge nodes fit within a single page.
	We use the notation $\node \in \page$ to express that \node is stored in page \page.
	The \BPtree nodes are stored next to each other in memory and they fit in $\pageAmount = k / \nodeAmount$ pages $\pageSet = \tuple{\pageEl{1}, ... \pageEl{\pageAmount}}$.

	Next, we will prove the security of Construction 1.
	The first step is to define the leakage functions that are based on the attack model described in Section~\ref{sec:highLevelDesign:subsec:attackModel}.
	\begin{description}[style=unboxed,leftmargin=.3cm, itemsep=-2pt,topsep=0pt]	
		\item[\leakEPara:] 
			Given the key-value pairs $\inputSet = ((\keyB{1}, \valueBEl{1}), ... (\keyB{n}, \allowbreak \valueBEl{n}))$, this function outputs the amount $n$ of values, the size of each value and the amount of \BPtree nodes \nodeAmount.
		\item[\leakHPara:] 
			Given the key-value pairs $\inputSet$, the plaintext $B^+$-tree \tree and the search range \range and given point in time \timePar{}, this function outputs the pages access pattern \pageAccessPatPara and the values access pattern \valueBAccessPatPara. 

			Loosely speaking, the pages access pattern \pageAccessPatPara is a tree that contains all pages in \pageSet that get accessed when the range \range is searched for. 
			An edge in \pageAccessPat from a parent to a child means that the child page gets accessed after the parent page.
			More formally, we define $M$ as the set that contains pages, in which leaf nodes are present that contain keys from the search range, \ie $M=\{ \page \, | \, \page \in \pageSet \wedge \node \in \page \wedge \nodeEl{} \inB \tree \wedgeB \nodeIsLeafVar{} \wedgeB \nodeKey{j} \inB \range, \, j \inB [1,\branchingF-1]\}$.
			Additionally, we define \nodeParent{}{1} as the parent node of \node in \tree and \nodeParent{}{l} denotes the node that is reached by moving $l$ layers up in the tree \tree starting from \node.
			Now, we can specify the node set \nodeSetT of \pageAccessPat as $\nodeSetT = \{ \pageEl{i} \, | \, \pageEl{i} \in M\} \bigcup_{} \{ \pageEl{i} \, | \, \pageEl{j} \in M \wedge \nodeEl{1} \in \pageEl{j} \wedge \pageEl{i} \in \pageSet \wedge \nodeEl{2} \in \pageEl{i} \wedge \nodeEl{2} \in \tree \wedge \nodeEl{2} == \nodeParent{1}{l}, l \in [1,\treeH - 1] \}$.	
			The edge set of \pageAccessPat is $\{(\pageEl{i}, \pageEl{j}) \, | \, \pageEl{i}, \pageEl{j} \in \nodeSetT \wedge \exists 
			\, \nodeEl{1}, \nodeEl{2} \in \tree : \nodeEl{1} \in \pageEl{i} \wedge \nodeEl{2} \in \pageEl{j} \wedge i \neq j \wedge \nodeEl{1} == \nodeParent{2}{1}\}$.
			The time parameter \timePar{} defines a snapshot of the random (but fixed) order of sibling nodes at a given point in time.
			See Figure~\ref{fic:pageAccessPattern} for an illustrative example.

			The values access pattern \valueBAccessPatPara is defined as the pointers to the result values together with the pages that contain nodes in which these pointers are stored.
			More formally, $\valueBAccessPatPara = \{ (\page, \pointerSetEl{\page}) \, | \, \page \in \pageSet \wedge \node \in \page \wedge \node \in \tree \wedge \nodeIsLeaf \wedge \exists \, \nodeKey{j} \in \range, j \inB [1, \branchingF-1] \wedge \pointerSetEl{\page} = \{ \nodePointer{l} \, | \, \node \in \page \wedge \nodeKey{l} \in \range, l \inB [1, \branchingF-1] \}$.
			The time parameter \timePar{} defines a random but fixed order of the pointers.
			
	\end{description}

	\begin{figure}[h]%
	  \centering%
	    \begin{subfigure}{\columnwidth}
	      \includegraphics[width=\columnwidth]{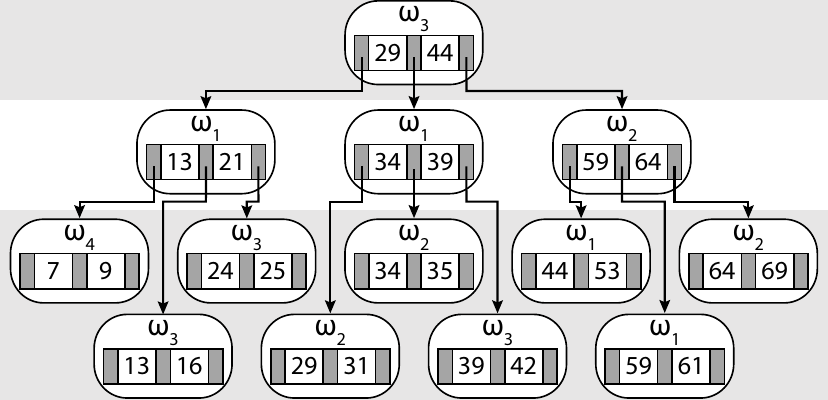}
	      \subcaption{}
	    \end{subfigure}
	  
	    \begin{subfigure}{0.49\columnwidth}
	      \includegraphics[width=\columnwidth]{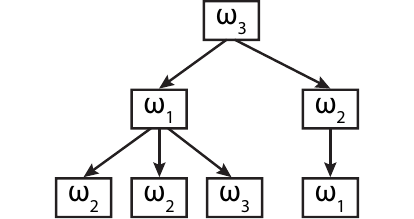}
	      \subcaption{}
	    \end{subfigure}
	    \begin{subfigure}{0.49\columnwidth}
	      \includegraphics[width=\columnwidth]{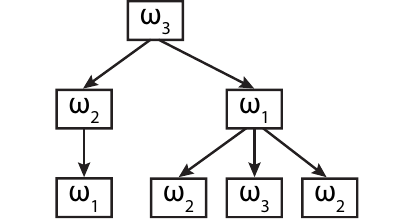}
	      \subcaption{}
	    \end{subfigure}
	    \vspace*{-.6em}
	    \caption{Illustration of page access pattern leakage: (a) example $B^+$-tree \tree (page containing the node on the top), (b) leakage \pageAccessPatPara for  $\range = [33,55]$ and $B^+$-tree \tree at \timePar{1}, (c) leakage \pageAccessPatPara for  $\range = [33,55]$ and $B^+$-tree \tree at \timePar{2}}%
	    \label{fic:pageAccessPattern}%
	\end{figure}%
	\vspace*{0.5em}

	\noindent
	The pages access pattern and the values access pattern are worst case estimations.
	An attacker would require many queries to exactly determine which page is a child of another page in \pageAccessPat.
	The same applies for the exact matching of result values to a page.
	\vspace{0.5em}

	\begin{thm}[Security]
		The secure hardware $B^+$-tree construction \HSBTO is \tuple{\leakE, \leakH}-secure according to Definition \ref{thm:sec}.
	\end{thm}

	\begin{proof}
		We describe a polynomial-time simulator \simulator for which a PPT adversary \adversary can distinguish between \real and \ideal with negligible probability.

		\begin{itemize}[leftmargin=0.4cm, itemsep=-2pt,topsep=0pt]
		\item 
			\textit{Setup:}
			\simulator creates a new random key $\SKR = \PSEGenPara$ and stores it.
		\item 
			\textit{Simulating \treeEnc:}
			\simulator gets \leakE and creates \nodeAmount nodes $\nodeSet = \tuple{\nodeEl{1}, ..., \nodeEl{\nodeAmount}}$ filled with random keys, random pointers and increasing node ids. 
			These nodes are stored in the pages $\tuple{\pageEl{1}, ... \pageEl{\pageAmount}}$.
			Additionally, \simulator generates $n$ encryptions of random values $\cipherSet = \tuple{\cipherEl{1}, ..., \cipherEl{n}}$ using \PSEEnc, the number of values and the size of the values. 
			Every encrypted value is given a distinct index.
			\simulator outputs $\treeEnc = (\nodeSet, \cipherSet)$

			All described operations are possible, because the amount $n$ of values, the size of each value and the amount of nodes \nodeAmount are included in the leakage. 			
			The simulated \treeEnc has the same size as the output \real and the IND-CPA-security of \PSE makes the nodes indistinguishable from the output of \real.
		\item 
			\textit{Simulating \token:}
			The simulator \simulator creates two random values (\randomEl{1} and \randomEl{2}) and encrypts them: $\tokenS = \PSEEncParaVar{\SKR}{\randomEl{1}}$, $\tokenE = \PSEEncParaVar{\SKR}{\randomEl{2}}$.
			\simulator outputs $\token = (\tokenS, \tokenE)$.
			
			The simulated \token is indistinguishable from the output of \real as a result of the IND-CPA-security of \PSE.
		\item
			\textit{Simulating secure hardware:}
			At time \timePar{}, the simulator \simulator receives a range token \token and \leakH.
			\simulator uses \pageAccessPatPara to simulate the page access pattern.
			For this task, \simulator starts at the root of \pageAccessPat and follows the links as unambiguously defined by \timePar{}.
			Afterwards, it outputs \pointerSet = $\bigcup_{w}$ \pointerSetEl{\page} for every $(\page, \pointerSetEl{\page}) \in $ \valueBAccessPatPara.
			The leakage \valueBAccessPat determines the order of the pointers in \pointerSet for a specific point in time \timePar{}.
			
			\adversary cannot distinguish between the page access of \real and the simulated access, because the page access pattern delivers deterministic results.
			Therefore, the results are consistent for different requests of the same range and also for queries of distinct or overlapping ranges.			
			Furthermore, the number of result pointers matches and the pointers are consistent, because \valueBAccessPatPara is unambiguous. 
			The values pointed on are indistinguishable, because they are protected by IND-CCA secure encryption.
		\end{itemize}

	\end{proof}

	Cloud computing enables the fast and cost-effective processing of large amounts of data.
	Unfortunately, Construction~1 suffers from the substantial problem mentioned before: the memory reserved for SGX is limited to \SI{128}{\mega\byte}, and only about \SI{96}{\mega\byte} can be used for data and code.
	SGX supports a larger enclave size, but enclave pages have to be swapped in and out in this case.
	Our evaluation (see Section \ref{sec:eval}) shows that even $50{,}000{,}000$ values are possible.
	A \BPtree, however, is only a small part of a full encrypted DBMS based on our constructions.
	Other components would occupy large regions of the restricted enclave memory and thus further limit the available space.	
	For that reason, we present a second construction in the next section that does not have this problems.	
	

\subsection{Construction~2}
\label{sec:design:subsec:const2}

	In this section, we describe our second correct (according to Definition~\ref{thm:cor}) and secure (according to Definition \ref{thm:sec}) construction.
	Instead of loading all nodes into the enclave the main idea is to only load the nodes required to traverse the tree.
	The challenge is to optimize the communication bottleneck between the untrusted part and the enclave.
	We performed extensive benchmarking and algorithm engineering in order to identify and minimize the most important run-time consuming tasks, such as switching between the untrusted part and the enclave.
	The decisive advantage of our second construction is that the required memory space inside the enclave is $O(1)$ for a tree of arbitrary size.
	The trade-off is that all nodes are stored encrypted inside main memory or on disk and thus have to be decrypted by the enclave.
	This also leads to a slightly larger leakage than in the first construction, namely a finer-granular access pattern on node instead of page level (details are described by a formal model and proof later in this section).

	The setup of the \HSBT-scheme is slightly different than in the first construction in order to implement the described features.
	As before, the \BPtree is constructed and encrypted at the client and is then transferred to the cloud provider.
	However, the application does not reserves memory for the whole \BPtree structure inside the enclave. 
	Instead, it only reserves a fixed space, denoted as \reservedSpace, for on the fly processing.
	The remote attestation and secure key deployment are performed as in the previous construction.

	\tuple{\HSBTTSetup,\:\allowbreak \HSBTTEnc,\:\allowbreak \HSBTTTok,\:\allowbreak \HSBTTDec,\:\allowbreak \HSBTTSearchRange,\allowbreak \HSBTTSearchRangeTrusted} is a second \HSBT-scheme consisting of six algorithms that utilizes a pseudo\-random permutation \ranPermDef.

	All algorithms but \HSBTTSearchRange and \HSBTTSearchRangeTrusted exactly match the descriptions of Construction~1.
	In contrast to Construction~1, these two algorithm get extended by the parameter \treeEnc and $nodes$, respectively, because it is required to pass the encrypted tree and the encrypted nodes inside the enclave.
	We now describe the modified algorithms:
	
	\begin{description}[style=unboxed,leftmargin=.3cm, itemsep=-2pt,topsep=0pt]
		\item[$\cipherSetRes \rnd \HSBTTSearchRangePara$:] 
			Take the search token \token and the encrypted tree \treeEnc as input.
			At the beginning, pass only the root node to the trusted part and receive pointers to nodes that should be traversed next.
			The trivial solution is to pass one node after another to the enclave.
			A severe problem with this design is that every context switch from the untrusted to the trusted part or back causes a substantial overhead.
			We therefore optimized the number of context switches: transfer as many nodes as currently in the queue, but not more than fit into \reservedSpace. 
			We denote the maximal number of nodes as \maxAmount, which is directly influenced by the space reserved inside the enclave during the setup process: $\maxAmount = \reservedSpace / (o \cdot \SI{128}{\bit})$ where $o$ is the number of AES-blocks used by each node.
			Nodes are passed until no further are requested.
			Then output \cipherSetRes by dereferencing pointers to the values.
			See Algorithm~\ref{algo:SHBTSearch2} for details.
		\item[$\pointerSet \rnd \HSBTTSearchRangeTrustedPara$:] 
			Take a search token \token and nodes \nodeSet as input.
			Every incoming node is encrypted with \skeyT.
			During the setup phase, \skeyT was deployed inside the secure hardware.
			Therefore, the algorithm is able to decrypt all nodes and the token.
			Then, search all keys falling in the query range, whereby \emph{all} keys are accessed.
		 	Finally, return the corresponding pointers in a random order.			
			See Algorithm~\ref{algo:SHBTSearchTrusted2} for details.
	\end{description}
	\vspace*{0.2em}

	\begin{algorithm}		
		\fontsize{8}{9.5}\selectfont%
		\begin{algorithmic}[1]
			\State \nodeSet = $\emptyset$ \Comment{FIFO queue}  

			\State \nodeSet.\Call{enqueue}{root}

			\State $results$ = $\emptyset$\;   

			\While{\nodeSet $\neq$ $\emptyset$}
				\For{i=0; i $<$ \nodeSet.\textit{size} $\&\&$ i $<$ maxAmount; \phantom . $i$++}
					\State $\nodeSet_{tmp}$ = \nodeSet.\Call{dequeue}{\phantom .}
				\EndFor

				\State $results_{tmp}$ = \Call{HSBT2\_SearchRangeTrusted}{\token, $\nodeSet_{tmp}$}

				\For{$($\textit{isLeaf}, $p$$)$ in $results_{tmp}$}
					\If{\textit{isLeaf}}
						\State $results$.\Call{add}{*$p$}
					\Else
						\State \nodeSet.\Call{enqueue}{*$p$}
					\EndIf
				\EndFor
			\EndWhile

			\State \textbf{return} $results$
		\end{algorithmic}
		\caption{\HSBTTSearchRangePara \label{algo:SHBTSearch2}}%
	\end{algorithm}
	\vspace*{-1em}

	\begin{algorithm}[ht]		
		\fontsize{8}{9.5}\selectfont%
		\begin{algorithmic}[1]
			\State \tokenP = \PSEDecParaVar{\skeyT}{\token}
			\State parse \tokenP as $\tuple{\rangeS, \rangeE}$

			\State $\nodeSet_{tmp}$ = \{\PSEDecParaVar{\skeyT}{\text{*X}_0}, \PSEDecParaVar{\skeyT}{\text{*X}_1}, ...\}

			\State \pointerSet = $\emptyset$
			
			\For {\node in $\nodeSet_{tmp}$}
				\If{not \node.\textit{isLeaf} and \rangeS $<$ \node.$k_1$}
					\State \pointerSet.\Call{add}{\nodePointer{0}}						
				\EndIf

				\For{$i = 1$, $i < \branchingF - 1$, $i$++}			
					\If{
						(\node.$k_i$ $\leq$ \rangeS $<$ \node.$k_{i + 1}$) $||$ (\node.$k_i$ $\leq$ \rangeE $<$ \node.$k_{i + 1}$) $||$ \newline \phantom a \phantom a \phantom a \phantom a \phantom a (\rangeS $\le$ \node.$k_i$ $\&\&$ \node.$k_{i + 1}$ $\le$ \rangeE)}  
							\State \pointerSet.\Call{add}{\nodePointer{i}}
					\EndIf	
				\EndFor

				\If{\rangeE $\geq$ \node.$k_{b-1}$}
					\State \pointerSet.\Call{add}{\nodePointer{b-1}}
				\EndIf
			\EndFor

			\State \pointerSet = random permutation of \pointerSet

			\State \textbf{return} $($\text{*P$_0$}.\textit{isLeaf}, P$_0)$, $($\text{*P$_1$}.\textit{isLeaf}, P$_1)$, ...
		\end{algorithmic}
		\caption{\HSBTTSearchRangeTrustedPara \label{algo:SHBTSearchTrusted2}}
	\end{algorithm}

	The construction is correct according to Definition \ref{thm:cor}, because it is based on a textbook \BPtree traversal. 
	The difference to the textbook algorithm is that the nodes are loaded inside the enclave after another and that each node is encrypted.
	These changes do not influence the correctness, because each node remains accessible to the enclave and the encryption (at the client) and the decryption (inside the enclave) are based on a correct \PSE-scheme.

	We prove the security of Construction 2 by again defining the leakage functions that are based on the attack model described in Section~\ref{sec:highLevelDesign:subsec:attackModel}.
	
	\begin{description}[style=unboxed,leftmargin=.3cm, itemsep=-2pt,topsep=0pt]	
		\item[\leakEPara:] 
			Given the key-value pairs $\inputSet = ((\keyB{1}, \valueBEl{1}), ... (\keyB{n}, \allowbreak \valueBEl{n}))$, this function outputs the amount $n$ of values, the size of each value and the amount of \BPtree nodes \nodeAmount.
		\item[\leakHPara:] 
			Given the key-value pairs $\inputSet$, the plaintext $B^+$-tree \tree, the search range \range and given point in time \timePar{}, this function outputs the nodes access pattern \treeAccessPatPara and the value pointers access pattern \valueBAccessPatPara. 

			The nodes access pattern \treeAccessPatPara is a tree that contains the storage positions of all nodes in \tree that get accessed when searching for the range \range.
			For a more formal definition, we denote the set of leaf nodes that contain keys from the range as $M$, \ie $M=\{ \nodeEl{} \, | \, \nodeEl{} \inB \tree \wedgeB \nodeIsLeafVar{} \wedgeB \nodeKeyVar{}{j} \inB \range, \, j \inB [1,\branchingF-1]\}$.
			We again denote \nodeParent{}{j} as the parent node of \node that is reached by moving $j$ layers up in the tree \tree starting from \node.
			Here, we denote a node that only contains the storage position of a node \nodeEl{i} as \nodeElT{i}. 
			Now, we can specify the node set \nodeSetT of \treeAccessPat as $\nodeSetT = \{ \nodeElT{i} \, | \, \nodeEl{i} \inB M\} \cup_{} \{ \nodeElT{i} \, | \, \nodeEl{i} \in \tree \wedge \node \in M \wedge \nodeEl{i} == \nodeParent{}{j}, j \in [1,\treeH - 1] \}$.	
			The set of directed edges in \treeAccessPat is $\{(\nodeElT{i}, \nodeElT{j}) \, | \, \nodeElT{i}, \nodeElT{j} \in \nodeSetT \wedge \exists \, \nodeEl{i}, \nodeEl{j} \in \tree \, : \,  \nodeEl{i} == \nodeParent{j}{1}\}$.
			The time parameter \timePar{} defines a snapshot of the random (but fixed) order of sibling nodes at a given point in time. 
			See Fig.~\ref{fic:treeAccessPattern} for an illustrative example.

			The value pointers access pattern \valueBAccessPatPara is defined as the pointers to the result values together with the leaf nodes in which these pointers are stored.
			More formally, $\valueBAccessPatPara = \{ (\nodeEl{}, \pointerSetEl{\node}) \, | \, \node \in \tree \wedge \nodeIsLeaf \wedge \exists \, \nodeKeyVar{}{j} \in \range, j \inB [1, \branchingF-1] \wedge \allowbreak \pointerSetEl{\node} = \{ \nodePointer{l} \, | \, \allowbreak \nodeKeyVar{}{l} \in \range, l \inB [1, \branchingF-1] \}\}$.
			The time parameter \timePar{} defines a random but fixed order of the pointers.
			
	\end{description}

	\begin{figure}[h]%
	  \centering%
	    \begin{subfigure}{\columnwidth}
	      \includegraphics[width=\columnwidth]{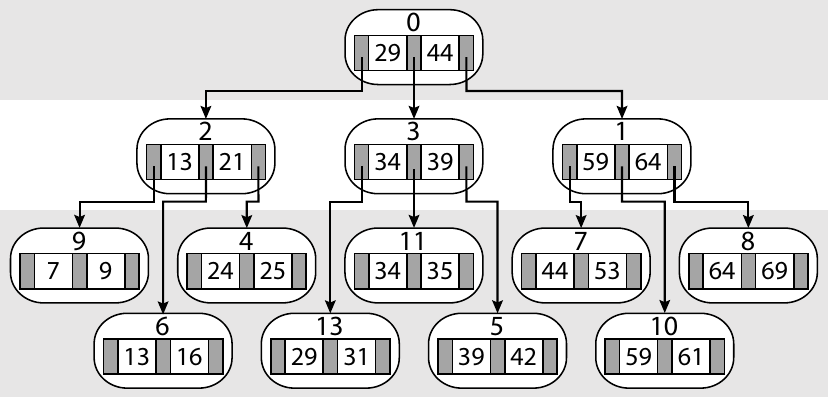}
	      \subcaption{}
	    \end{subfigure}
	  
	    \begin{subfigure}{0.49\columnwidth}
	      \includegraphics[width=\columnwidth]{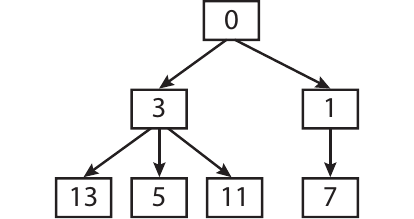}
	      \subcaption{}
	    \end{subfigure}
	    \begin{subfigure}{0.49\columnwidth}
	      \includegraphics[width=\columnwidth]{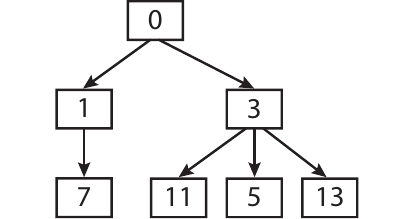}
	      \subcaption{}
	    \end{subfigure}
	    \vspace*{-.6em}
	    \caption{Illustration of nodes access pattern leakage: (a) example $B^+$-tree \tree (random storage position on the top), (b) leakage \treeAccessPatPara for  $\range = [33,55]$ and $B^+$-tree \tree at \timePar{1}, (c) leakage \treeAccessPatPara for  $\range = [33,55]$ and $B^+$-tree \tree at \timePar{2}}%
	    \label{fic:treeAccessPattern}%
	\end{figure}%
	\vspace{0.5em}

	\begin{thm}[Security]
		The secure hardware $B^+$-tree construction \HSBTT is \tuple{\leakE, \leakH}-secure according to Definition \ref{thm:sec}.
		\label{thm1}
	\end{thm}
	\vspace*{0.3em}

	\begin{proof}
		The simulator \simulator works as follows:

		\begin{itemize}[leftmargin=0.4cm, itemsep=-2pt,topsep=0pt]	
		\item 
			\textit{Setup:}
			\simulator creates a new random key $\SKR = \PSEGenPara$ and stores it.
		\item 
			\textit{Simulating \treeEnc:}
			\simulator gets \leakE and creates \nodeAmount nodes filled with \SI{32}{\bit} random keys and random pointers.
			Every node \nodeEl{i} gets a unique node id $\nodeIdEl{i} \in [0, \nodeAmount]$ and the node is stored at the position $\nodePos{i} = \ranPermPara{\SKR}{\nodeIdEl{i}}$.
			The nodes get encrypted with \PSEEncParaVar{\SKR}{\nodePos{i}}.
			Additionally, \simulator generates $n$ encryptions of random values $\cipherSet = \tuple{\cipherEl{1}, ..., \cipherEl{n}}$ using \PSEEnc, the number of values and the size of the values. 
			Every encrypted value is given a distinct index.
			\simulator outputs $\treeEnc = (\nodeSet, \cipherSet)$.

			All described operations can be executed by \simulator, because the information required for the encryption of values is included in the leakage.		
			The simulated \treeEnc has the same size as the output of \real and the IND-CCA security of \PSE makes the nodes and values indistinguishable from the output of \real.
		\item 
			\textit{Simulating \token:}
			The simulator \simulator creates two random values \randomEl{1} and \randomEl{2} and encrypts them: $\token  \rnd \PSEEncParaVar{\skeyT}{\rangeS || \rangeE}$.
			\simulator outputs $\token$.
			
			The simulated \token is indistinguishable from the output of \real as a result of the IND-CCA security of \PSE.
		\item
			\textit{Simulating secure hardware:}
			At time \timePar{}, the simulator \simulator receives encrypted nodes (denoted as \nodeSet), a token \token and \leakH.
			It has to simulate the output of the secure hardware enclave.
			The simulator decrypts every node $\nodeEl{i} \in \nodeSet$ with \PSEDecParaVar{\SKR}{\nodeEl{i}}.
			We differentiate between two cases for every $\nodeEl{i}$:
			\begin{enumerate}[style=unboxed,leftmargin=.3cm,topsep=0pt]
				\item
					\nodeEl{i} is not leaf:
					\simulator reads the id of \nodeEl{i} and searches the corresponding \nodeElT{i} in \treeAccessPatPara.
					It returns a pointer to all children in the order defined by \timePar{}.
					
					\adversary cannot distinguish between the output of \real and the simulated output, because the pointers point to indistinguishable nodes according to the IND-CCA security of \PSE. 
					Furthermore, the results are consistent for different requests of the same range as the nodes access pattern delivers deterministic results and the pseudorandom permutation creates unambiguous positions for the simulated nodes.
					The same argument applies for queries of distinct or overlapping ranges.
				\item
					\node is leaf: \simulator uses the leakage \valueBAccessPatPara to output all result pointers \pointerSet = $\bigcup_{\node}$ \pointerSetEl{\node}, $\forall$ $(\node, \pointerSetEl{\node}) \in $ \valueBAccessPat in the order defined by \timePar{}.
					
					This output is indistinguishable from the output of \real as the number of result pointers matches and the pointers are consistent because \valueBAccessPatPara is unambiguous. 
					The values pointed on are indistinguishable, because they are protected by IND-CCA secure encryption.
			\end{enumerate}
		\end{itemize}
	\end{proof}
	
	The main difference in the leakages of Construction~1 and Construction~2 is the granularity of the tree and value pointers access pattern. 
	In the second construction, the attacker is able to reveal accesses on a node level. 
	In contrast, the attacker is able to reveal accesses on a page level in the first construction, because SGX inherently leaks the page access pattern.


\subsection{Active attacker}
\label{sec:design:subsec:malAtt}
	Construction~1 and 2 are secure and correct for a passive attacker, but our overall goal is to enable the outsourcing of data to untrusted cloud providers.
	Therefore, it is important to consider active attackers.
	This attacker type tries to thwart the correctness and tries to gain additional sensitive information by not following the defined protocol.
	We omit concrete definitions of correctness and CKA2-HW-security under a active attacker, because they are easily deducible from Definition~\ref{thm:cor} and Definition~\ref{thm:sec}.
	In the following, we consider only Construction~2, but the arguments and techniques can be applied to Construction~1 with minor modifications. 

	\textbf{Attack vectors.}
	We identified two basic attack vectors that cover a wide range of possible attacks.
	Firstly, the attacker can try to attack the protection mechanisms of SGX to gain insights about data and algorithm execution not under his control.
	We rely on SGX's protection mechanism that guarantees security and correctness under an active attacker.
	However, as defined in our attacker model (see Section~\ref{sec:highLevelDesign:subsec:attackModel}), we consider various side channels.
	Protection mechanisms against these are implementation details and thus described in our implementation section (Section~\ref{sec:impl}).
	Secondly, the active attacker can try to influence the data and protocol execution that is under his control.
	Namely, all encrypted nodes, encrypted values, encrypted tokens and the \HSBTTSearchRange algorithm.
	He can do this to gain additional sensitive information or to prevent the client from getting the correct result.	
	In this section, we present arguments and protection mechanisms that address this second attack vector.

	For now, we assume that there is a mechanism that guarantees the following to the client: if it gets a result, all results in the response match the query and the response contains all matching results.
	This mechanism is presented later in this section, but we first consider the correctness and security implications of this assumption. 
	We show that there are no security implications by introducing an active attacker instead of an passive attacker.

	\textbf{Unprotected \emph{static data}.} 
	The only static data influenceable by the active attacker are values and nodes.
	The security of this data is guaranteed by using an authenticated IND-CPA secure encryption scheme. 
	The authenticated encryption also thwarts attacks on the correctness that try to modify or add static data, because these actions are noticed by the decryption algorithm.
	We do not consider the deletion of values or nodes, because it would lead to incomplete results and thus contradict the assumption stated above.

	\textbf{Unprotected \emph{dynamic data}.} 
	The only dynamic data influenceable by the attacker is the search token.
	Again, the security of tokens and the prevention of modifications and additions is provided by the authenticated IND-CPA secure encryption scheme.
	A replay attack does also not give any additional information to the attacker as the tree is static and the leakage stays the same for a replayed token.
	The correctness could only be influenced by a denial of service (DoS) attack by dropping tokens, but DoS is out of scope according to our attacker model.

	\textbf{Unprotected algorithms.} 
	\HSBTTSearchRange is the only remaining untrusted part under control of the active attacker.
	There is only a fixed set of deviations from the protocol that do not lead to a DoS.
	As explained before, all additional or modified static or dynamic data is directly noticeable by the enclave, which can simply reject further processing.
	Passing less nodes than possible to the enclave only slows down the process, but does not lead to additional information and does not impact the correctness.\footnote{Passing one node after another is already covered by the defined leakage.}
	The only remaining deviations are: (1) do not pass the root node first to the enclave, (2) pass the wrong nodes to the enclave (3) do not pass all requested nodes to the enclave and (4) do not pass all results to the client.
	It is important to note that all these deviations do influence the correctness of the protocol, but not the security.
	We now first present some straightforward, but problematic alternatives.
	Afterwards, we explain our mechanism that ensures the stated assumption and defends against the four possible deviations.

	One straightforward way to protect against deviations (2) - (4) is the usage of a Merkle tree structure (like done in \cite{lu_privacy-preserving_2012}).
	It would be possible by transferring all values of a leaf node that contains any matching value, the hashes of the matching nodes and the hashes of any sibling of all matching nodes.
	Together with a stored root hash, the client can check the hashes bottom-up.
	However, this requires to transfer more data than necessary, leaks unnecessary information by touching sibling nodes and requires additional processing at the client. 

	A further method to prevent against deviations (2) and (3) is to store a list of requested node ids inside the enclave.
	For every incoming node, the enclave removes the corresponding id from the list and can be sure that it received all and only correct nodes when the list is empty.
	However, this approach requires $O(r)$ storage inside the enclave for the total result size $r$.

	Deviations (4) can be mitigated by using the secure channel created during provisioning to transfer the results.
	Technically, it is possible, because the enclave has direct access to the values.
	Nevertheless, it would also require $O(r)$ additional storage to memorize $r$ result pointers.
	Additionally, it would not give any additional protection, because the data access are perceivable by untrusted software.

	We now present modifications to Construction~2 that do not exhibit the drawbacks of the Merkle tree approach and do not require $O(r)$ additional storage in the enclave.
	We utilize so called multiset hashes \cite{clarke_incremental_2003} as a building block to prevent the deviations (2) - (4) with only $O(1)$ additional storage.
	The properties of a multiset hash are: multiple values are hashed to a fixed-size bit string, elements can be added incrementally and efficiently, the order of the input can be arbitrary and there is an efficient equality check for two multiset hashes.
	Two multiset hashes are equal if the hash was calculated over same elements (independent of the order).
	As a concrete example, MSet-XOR-Hash from \cite{clarke_incremental_2003} can be used as a multiset hash.
	We refer to \cite{clarke_incremental_2003} for details and a security proof that is based on the hardness of breaking the underlying pseudorandom functions.

	The construction of the \BPtree is changed in the following way: the root node and every inner node \node store the child id \nodeIdCh{i} next to the pointers \nodePointer{i}; every leaf node \node stores a hash \nodeMsgHash{i} of the plaintext values next to the pointers \nodePointer{i}.
	The id of the root node $root.id$ is sent to the enclave together with \skeyT in the secure channel established during the setup phase.
	
	In the following modifications, we abbreviate the two algorithms of Construction~2 (\ie \HSBTTSearchRange and \HSBTTSearchRangeTrusted) as $\mathtt{Untrusted}$ and $\mathtt{Trusted}$.
	For every node \node that $\mathtt{Trusted}$ receives, it decrypts the node and checks if \nodeId matches $root.id$. 
	If this is the case, it creates a new nonce.
	Otherwise, it expects to receive a nonce from $\mathtt{Untrusted}$.
	The processing aborts if no nonce is present and any node but the root node was passed.
	This thwarts deviation (1).

	For every nonce, $\mathtt{Trusted}$ stores how many nodes it expects at the moment in $expectedNodes\-Amount$ and it stores three multiset hashes: $expectedNodes\-Hash$, $receivedNodes\-Hash$ and $result\-Values\-Hash$.
	For the root and every inner node \node, it adds \nodeIdCh{i} for every expected child to $expected\-Nodes\-Hash$.
	For each received node \node, it adds \nodeIdEl{} to $receivedNodes\-Hash$.
	It then processes the nodes as before.
	In each leaf node \node, it checks if \nodeKey{i} falls into the range.
	\nodeMsgHash{i} is added to $resultValuesHash$ if this is the case.
	Finally, it passes the list of pointers and the nonce to $\mathtt{Untrusted}$.

	$expectedNodes\-Amount$ is reduced for every incoming node. 
	After reaching $expectedNodes\-Amount = 0$ and adding the last ids of each received node to $receivedNodesHash$, $\mathtt{Trusted}$ compares $expected\-Nodes\-Hash$ with $received\-Nodes\-Hash$.
	If these hashes do not match, it did not receive the correct nodes from $\mathtt{Untrusted}$, aborts and deletes the nonce.
	Thus, the algorithm guarantees that it traversed all nodes that might contain an eligible result and thereby protects against deviation (2)\footnote{The attacker does not gain additional information from the processing of wrong nodes. 
	The reason is that the algorithm would not request them and thus already leaks that they do not contain values from the range.} and (3).
	Otherwise, it adds \nodeMsgHash{i} to $resultValuesHash$ for every result found in the last search round and returns the pointers together with an HMAC over $resultValuesHash$.
	$\mathtt{Untrusted}$ adds the HMAC to the client response.
	After the decryption, the client calculates a multiset hash over the received values, creates an HMAC and compares it with the received HMAC.
	The HMAC comparison protects against deviation~(4).

	Summarizing, the client is guaranteed the stated assumption through the presented mechanism.
	Performance measurements showed that the described protection mechanisms introduce an overhead of about $0.3ms$ at a query result size of $100$.

\subsection{Multiple users}
\label{sec:design:subsec:exten}

	So far, we considered a setup comprising one user, but multiple user directly supported by \name.
	Multiple users are able to concurrently query data without limitations, as concurrent tree traversals do not influence each other.
	The only requirement is that each user has access to the key \skeyT to create query tokens and \skeyM to decrypt the result.
	It is also possible that each user shares a different key \skeyT with the enclave.
	This would hide the search pattern of one user from all other users, but it requires a small modification in the protocol: the token has to be accompanied by client information, because the enclave has to identify the key to use for the token decryption.
	The nodes can be encrypted by any key that is known to the enclave.
	Particularly, it is not required to be a key shared with any user.


\section{Implementation}
\label{sec:impl}

Subsequently, we elaborate on implementation details, which are important with respect to performance or security.

\textbf{Platform.}
Our \name implementation is based on Intel's Software Development Kit (SDK) for SGX.
Services like attestation are available through the interfaces of the SDK.
The hardware platform is described in the evaluation (see Section~\ref{sec:eval}).

\textbf{Provisioning.} 
As described before, the client initially provisions \skeyT to the enclave.
This has to be done securely since \skeyT should not be revealed to any untrusted party.
We use the \emph{attestation} feature of SGX to establish a secure channel between the client and the enclave.
As described in Section~\ref{sec:overview:subsec:sgx}, the initial enclave memory content is measured during its creation. 
The first operation performed by the enclave is the creation of a key pair (secret key $sk_E$ and public key $pk_E$).
The randomness required for the key generation is provided by the hardware random number generator (\texttt{rdrand}~\cite{intel_manual}) available in current CPUs.
As the enclave's memory is isolated, $sk_E$ is never revealed to any other party.

Next, the enclave sends its just created $pk_E$ to the quoting enclave (QE).
The QE creates an signature over both, the measurement of the initial memory content of the enclave ($M_E$) (\ie the code and static data) and the public key: $\sigma_{QE}(M_E || pk_E)$.
$M_E$, $pk_E$ and $\sigma_{QE}(M_E || pk_E)$ are sent to the client who verifies the signature (given Intel's public key).
The client also verifies the measurement of the initial enclave state to establish trust into this enclave instance (\ie only if the measurement matches a known value she will continue).
The client then encrypts \skeyT with $pk_E$ and sends it back to the enclave, which is the only entity that can decrypt \skeyT.
Subsequently enclave and client share \skeyT, which they use for secure communication.

\textbf{Side channels.} 
	Our implementation is concerned with three types of (side) channels: external resource access, page-fault side channel and cache timing side channel (see Section~\ref{sec:highLevelDesign:subsec:attackModel}).
	In particular, by means of all three channels an adversary can observe access patterns to memory with the goal of inferring sensitive information from the observed access patterns.

	In \name, access to external resources by the enclave is limited to \BPtree nodes.
	Here, the attacker's goal is to leak information about the tree structure and ultimately the order of values stored in the database and searched for by the client. 
	Our Construction~2 explicitly covers the access to external resources, \ie nodes, in its leakage.

	While the external resources used by \name store only data, page-fault side channel and cache timing side channels allow to observe access to enclave memory which contains both, data and code.
	These side channels do not reveal sensitive data directly.
	However, they might reveal access to memory locations, which can reveal sensitive information if the memory access to code or data differs depending on sensitive data.

	As before, the attacker's goal might be to extract information about the tree structure and the order of values stored and searched.
	Additionally the attacker might try to extract cryptographic secrets from the enclave.
	For instance, the attack might aim at learning the secret key \skeyT used to encrypt tree nodes or learning one of the keys used during the establishment of a secure channel in the initial provisioning phase.

	The page-fault side channel allows the attacker to reliably observe memory access patterns, however, the granularity is relatively coarse (\SI{4}{\kilo\byte}).
	All accesses within the same page are indistinguishable for the attacker and, thus, are not exploitable.
	Construction~1 explicitly considers the leakage of the tree structure through this side channel. 
	In Construction~2 the page-fault side channel does not leak additional information, as nodes are smaller than memory pages and the nodes access pattern is leaked anyway by storing the \BPtree outside of the enclave.

	Cache timing side channel allow finer grained memory access observations while being less reliable. 
	Nevertheless, assuming an adversary who is able to observe accesses within a node, the attacker needs to determine which links to child nodes are followed.
	Our algorithm, however, accesses every key and pointer, whether the pointer is followed or not.
	By this and other fine grained implementation details, we achieve data independent accesses and thwart the cache timing side channel.

	Leakage of cryptographic keys are thwarted for page-fault and cache timing side channel by using leakage resilient implementations and hardware features~\cite{Bernstein_2012}. 
	For instance, the AES implementation used in \name uses AES-NI hardware which holds the S-Boxes in CPU registers instead of RAM, thus hampering cache side channel attacks~\cite{X10,MKS12}.

\textbf{Memory Management.}
We implemented both constructions of \name.
In particular Construction~2 is optimized with respect to memory transfer operations, and context switches between untrusted and trusted part.
To reduce the number of context switches a list of requested nodes is hold by the untrusted part.
Nodes from this list are transfered and processed at once, \ie with only one switch.
The memory transfer is optimized by exploiting the fact that the enclave can access the memory of its host process.
The \BPtree is loaded in the host process' memory from where the enclave can fetch nodes directly, decrypt them and process them.
This is much more efficient than copying nodes explicitly into enclave memory before decrypting and processing them.


\section{Performance Evaluation}
\label{sec:eval}

In this section, we present our evaluation results collected in a number of experiments.
First, we compare our two constructions described in Section~\ref{sec:design}.
Then, we examine the effects of the different ways of memory management for the constructions.
Finally, we compare our solution against the currently fastest polylogarithmic range query search algorithm presented in~\cite{demertzis_2016, faber2015rich}. 

\textbf{Platform.} Our evaluation system was equipped with an Intel Core i7-6700 processor at $3.40\,\unit{Ghz}$ (\ie sixth generation Core-i7, code name Skylake) and $32\,\unit{GB}$ DDR4 RAM.
64-bit Ubuntu 14.04.1 extended with SGX support (\eg EPC management) was used as operating system.
We evaluate our implementation on a real SGX hardware and provide actual measurements.

\subsection{Construction 1 vs.~Construction~2}
\label{sec:eval:subsec:constComp}

First, we compare the performance of our two constructions.
The parameters of the \BPtree are held constant for this comparison: the branching factor is $10$ and the tree contains $1{,}000{,}000$ key-value pairs.
Queries with five different sizes of the result set are used: $2^0, 2^4, 2^8, 2^{12}, 2^{16}$.
The search ranges were selected uniformly at random and every result size is tested with $1{,}000$ different ranges.
Figure \ref{fig:constructionsCompare} depicts the results of this evaluation, whereby the x-axis shows the size of the result set and the y-axis shows the median of the run-times in ms.

\begin{figure}[h]%
	\centering%
	\includegraphics[width=\columnwidth]
		{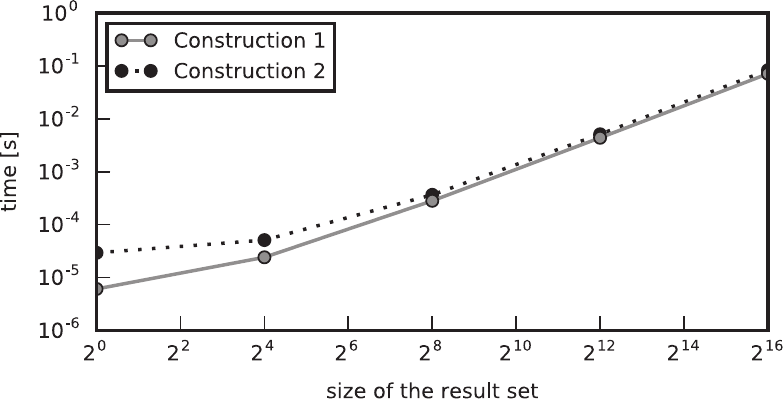} 
	\caption{Comparison of constructions}
	\label{fig:constructionsCompare}
\end{figure}%

The performance difference can be explained by the following effects:
\begin{itemize}[leftmargin=0.4cm, itemsep=-2pt,topsep=0pt]
	\item \textbf{Processor mode switch.} Before executing inside an enclave, the processor has to switch into ``enclave mode''.
	This includes, \eg storing the current CPU context on the host process' stack and loading the CPU context of the enclave. 
	In Construction~1 only one switch is required, whereas in the Construction~2 $O(\log_b {n})$ switches are performed, as at least each level of the \BPtree is loaded into the enclave.

	\item \textbf{Data transfer.} In Construction~1, the data transfer between trusted and the untrusted code is limited to the result set and the query whereas in Construction~2 also part of the \BPtree is transferred between the two components.

	\item \textbf{Access to plain data.} In Construction~1, decryption is a one-time effort after loading the entire \BPtree into the enclave.
	During query processing, it has access to plaintext nodes of the \BPtree.	
	Construction~2 incrementally loads the \BPtree nodes from untrusted storage. 
	All processed nodes need to be decrypted during query processing.
\end{itemize}
\vspace*{0.4em}

Construction~2, therefore, is slower than Construction~1 by a small factor at any result size.
For an increasing size of the result set, both algorithms search a linearly increasing part of the tree.	
Figure~\ref{fig:constructionsCompare} shows that the run-times of our two constructions converge (on a logarithmic scale).
This shows that the effects described above diminish compared to the search time of the algorithm.
	
\subsection{Memory Management}
\label{sec:eval:subsec:constCharacs}
	
	In order to identify the limiting parameters in the memory management of our two constructions, we evaluate \BPtrees with different tree sizes (amounts of key-value pairs) and branching factors.
	On each tree we ran $1{,}000$ randomly chosen queries with result set size of 100 and tested with the branching factors 10, 25, 50 and 100.
	The results of these evaluation are depicted in Figure~\ref{fig:bfCompare1} and Figure~\ref{fig:bfCompare2}.
	The x-axis shows the size of the \BPtree and the y-axis shows the median run-time of the queries.

	\begin{figure}[th]
		\centering
		\begin{subfigure}{\columnwidth}
		  \centering
		  \includegraphics[width=\columnwidth]{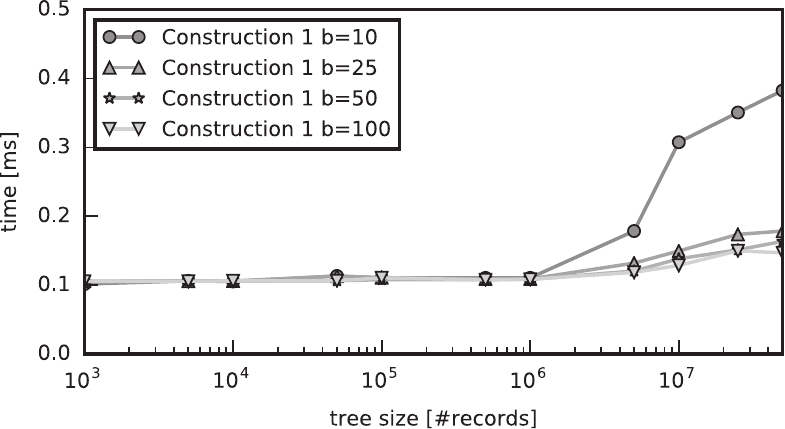}
		  \caption{}
		  \label{fig:bfCompare1}
		\end{subfigure}
		\begin{subfigure}{\columnwidth}
		  \centering
		  \includegraphics[width=\columnwidth]{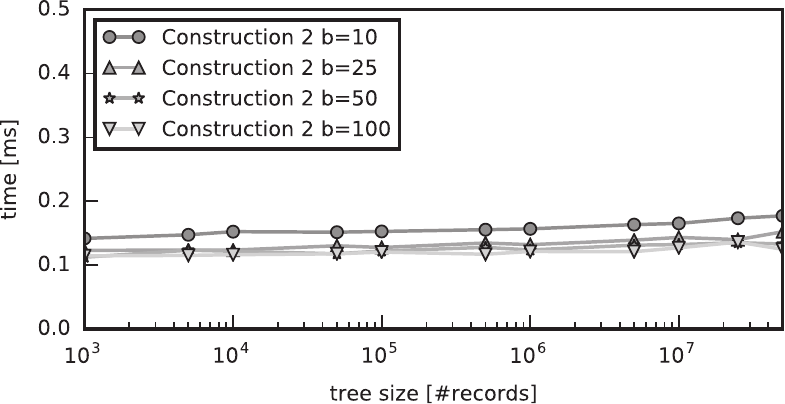}	  
		  \caption{}
		  \label{fig:bfCompare2}
		\end{subfigure}
		\caption{Effect of different branching factors in (a) Construction~1 and (b) Construction~2}
		\label{fig:bfCompare}
	\end{figure}

	In Figure~\ref{fig:bfCompare1}, we see a sharp increase of the run-time above a tree size of $10^6$ records.
	This is due to the exhausted memory in SGX and the virtual memory mechanism of the operating system that swaps pages in and out.
	This is not security critical, since pages remain encrypted and integrity protected by the SGX system, even when they are swapped out of the SGX protected memory.

	We see a significant difference in the impact of paging between the different branching factors.
	This becomes clear by considering the number of required page swaps.
	The lower the branching factor, the higher the number of nodes in a $B^+$-tree.
	The higher the number of nodes, the higher the number of accesses to different memory pages.
	The higher the number of different page accesses, the higher the probability of a swapped out page.

	In Figure~\ref{fig:bfCompare2}, we see that Construction~2 is not affected by paging, albeit supporting an unlimited tree size.
	Our data also shows that, as expected, higher branching factors result in better performance. 
	Disregarding the paging problem of Construction~1 above a tree size of $10^6$ records, a direct comparison of the constructions reveals that the runtime of Construction~2 approaches the runtime of Construction~1 for higher branching factors.

\subsection{Comparison with related work}
\label{sec:eval:subsec:searchEnc}

In this section, we compare our Construction~2 against the currently fastest approach with comparable security features and a security proof presented by Demertzis et al. in~\cite{demertzis_2016}.
The authors present seven different constructions that support range queries. 
The constructions have different tradeoffs regarding security, query size, search time, storage and false positives.
We do not compare against the highly secure scheme with prohibitive storage cost and also exclude the approaches with false positives as our construction does not lead to false positives.
Instead, we compare against the most secure approach without these problems: Logarithmic-URC.

We assume that the OXT construction from \cite{cash2013highly} is used as underlying symmetric searchable encryption scheme (SEE) by Logarithmic-URC.
Fundamentally, the SSE scheme is changeable, but the authors of~\cite{demertzis_2016} also utilize OXT for the security and performance evaluation.
One has to note that a quite equal construction as Logarithmic-URC was presented independently by Faber et al. in~\cite{faber2015rich}.
We implemented the algorithm of~\cite{demertzis_2016}, but a security and performance comparison to~\cite{faber2015rich} would lead to comparable results.

Table \ref{tbl:luComparison} compares our Construction~2 and Logarithmic-URC.
In this evaluation, we use a branching factor of 100 for Construction~2 and search for a randomly chosen range that contains 100 results.
Every test for the four different tree sizes ($100$, $1{,}000$, $10{,}000$, $100{,}000$) was performed $1{,}000$ times and the table shows the mean.
\vspace*{-0.3em}

\begin{table}[hbt]%
\centering%
\begin{tabular}{m{5em}||c|c|c|c}%
Tree Size 						& $100$ 		& $1{,}000$ 	& $10{,}000$ 	& $100{,}000$ \\%
\hline%
\hline%
Logarithmic-URC 				& \SI{0.015}{\second} 	& \SI{0.020}{\second} 	& \SI{0.051}{\second} 	& \SI{1.052}{\second} \\%
\hline%
Const.~2 \newline ($\branchingF = 100$)	& \SI{0.119}{ms} 	& \SI{0.121}{ms}	& \SI{0.124}{ms}	& \SI{0.125}{ms}\\%
\end{tabular}%
\caption{Time comparison of random range queries with Logarithmic-URC \cite{demertzis_2016} and our Construction~2}%
\label{tbl:luComparison}%
\end{table}%
\vspace*{-0.6em}

Our construction runs in about a tenth of a millisecond and with very moderate increase for all tree sizes.
In contrast, Logarithmic-URC requires at least multiple milliseconds up to a seconds for bigger trees.
A reason for the performance difference might be that OXT construction itself is less efficient then our construction.
Furthermore, the search time of OXT depends on the number of entries.
Logarithmic-URC fills the OXT construction with elements from a binary tree over the domain for every stored key.
An increasing domain severely increases the tree height of a binary tree and thus the number of entries for OXT.
In contrast, the height of the \BPtree in our construction increases much slower with the tree size.

A functional difference between Logarithmic-URC and Construction~2 is that Logarithmic-URC requires to fix the search key domain beforehand.
Cover a huge domain does negatively influence the setup and search time. 
In our construction, it is not necessary to fix the domain and the domain size has no performance implications.

It is not trivial to compare Logarithmic-URC and Construction~2 regarding security. 
The access pattern leakage and the leakage of the internal data structure of Logarithmic-URC is comparable to our access pattern leakages.
However, Logarithmic-URC additionally leaks the domain size, the search range size and the search pattern.
The search pattern reveals whether the same search was performed before, which might be sensitive information.
Furthermore, our construction only requires index storage in $O(n)$ instead of in $O(n \log D)$ as in Logarithmic-URC.	

\section{Related Work}
\label{sec:relWork}

\begin{table*}[t]
\footnotesize
\centering
\begin{tabular}{l||c|c|c|c|c}
Scheme 											& Search time 			& Query Size 	& Storage Size 	& Search Pattern Leakage & Order Leakage \\
\hline
\hline
Boneh, Waters \cite{BW07} 						& $O(n D)$ 				& $O(D)$ 		& $O(n D)$ 		& yes 					& no \\
\hline
Shi et al.~\cite{shi_multi-dimensional_2007} 	& $O(n \log D)$ 		& $O(\log D)$ 	& $O(n \log D)$ & yes 					& no \\
\hline
Shen et al.~\cite{shen_predicate_2009} 			& $O(n \log D)$ 		& $O(\log D)$ 	& $O(n \log D)$ & no 					& no \\
\hline
Lu \cite{lu_privacy-preserving_2012} 			& $O(\log n \log D)$ 	& $O(\log D)$ 	& $O(n \log D)$ & no 					& yes \\
\hline
\begin{tabular}{@{}l@{}}Demertzis et al. ~\cite{demertzis_2016} \\ Faber et al.~\cite{faber2015rich}\end{tabular}
 		& $O(\log R)$ 	& $O(\log R)$ 	& $O(n \log D)$ 	& yes 		& no \\
\hline
Poddar et al.~\cite{poddar_arx_2016} 			& $O(\log n)$ 			& $O(\log n)$ 	& $O(n)$ 		& no 					& yes \\
\hline
This paper 										& $O(\log n)$ 			& $O(1)$ 		& $O(n)$ 		& no 					& no \\
\end{tabular}
\caption{Comparison of range-searchable encryption schemes. $n$ is the number of keys, $D$ is the size of the plaintext domain and $R$ is the query range size.}
\label{tbl:comparison}
\vspace*{-0.4cm}
\end{table*}

Our work is related to searchable encryption, encrypted databases and secure implementations based on a TEE (\eg SGX).

\subsection{Searchable encryption}
\label{sec:relse}

Song et al.~introduced in \cite{song_practical_2000} the first searchable encryption schemes for single plaintexts.
In order to improve performance, Goh \cite{goh_secure_2003} and Curtmola et al.~\cite{curtmola_searchable_2006} introduced encrypted (inverted) indices.
However, these encryption schemes can only search for keyword equality and not ranges.

Searchable encryption scheme supporting range queries are rare.
Table~\ref{tbl:comparison} shows a comparison of different searchable encryption schemes and other schemes that support range queries.
Note that all existing range-searchable encryption schemes leak the access pattern -- including ours.
The first range-searchable scheme by Boneh and Waters in \cite{BW07} encrypt every entry linear in the size of the plaintext domain.
The first scheme with logarithmic storage size per entry in the domain was proposed by Shi et al.~in \cite{shi_multi-dimensional_2007}.
Their security model (match-revealing) is somewhat weaker than standard searchable encryption.
The construction is based on inner-product predicate encryption which has been made fully secure by Shen et al.~in \cite{shen_predicate_2009}.
All of these schemes have linear search time.

Lu built the range-searchable encryption from \cite{shen_predicate_2009} into an index in \cite{lu_privacy-preserving_2012}, thereby enabling polylogarithmic search time.
However, his encrypted inverted index tree reveals the order of the plaintexts and is hence only as secure as order-preserving encryption.
Wang et al.~\cite{wang_2014} proposed a multi-dimensional extension of Lu~\cite{lu_privacy-preserving_2012}, but it suffers from the same problem of order leakage.
There is no known searchable encryption schemes for ranges - until ours - that has polylogarithmic search time and leaks only the access pattern.

A Lu implementation done by us requires several seconds or minutes for a single range search, even with a security parameter much weaker than ours.
Hence, we not only improve asymptotic search time, but more importantly reduce the constants in order to open application to much larger data sets.

ORAM can in principle be used to hide the access pattern of searchable encryption.
However, Naveed shows that the combination of the two is not straightforward \cite{Nav15}.
Special ORAM techniques, like TWORAM \cite{GarMoh16}, are needed.

\subsection{Encrypted Databases}
\label{sec:reldb}

Encrypted databases, such as CryptDB \cite{popa_cryptdb:_2011}, use property-preserving encryption for efficient search.
Property-preserving encryption has very low deployment and runtime overhead due to the ability to use internal index structures of the database engine in the same way as on plain data.
Order-preserving encryption \cite{agrawal_order_2004,boldyreva_ope_2009,boldyreva_ope_2011,kerschbaum_optimal_2014} allows range queries on the ciphertexts as on the plaintexts.
With order-revealing encryption \cite{chenette2016practical,lewi2016order} a generalization of order-preserving encryption has been published recently.  
However, Naveed et al.~\cite{naveed_inference_2015} initiated the research direction of practical ciphertext-only attacks on property-preserving encryption, in particular order-preserving encryption, which recover the plaintext in many cases with very high probability (close to $100\%$) and further attacks followed~\cite{durak2016else,grubbs2016leakage}.

There have been a number of attempts to build indices for range queries based on deterministic encryption.
Bucketization of ciphertexts \cite{HMT04} groups ciphertexts on the server and filters results at the client.
Wang et al.~\cite{wang_secure_2013} uses distance-revealing encryption in order to build an r-tree.
Li and Omiecinski \cite{LO05} use prefix-preserving encryption in order to build a prefix tree for range searches.
However, all of these approaches are susceptible to the same attacks (and worse) as those by Naveed et al.

Four further approaches for a secure DBMS allowing range query evaluation have been published:
Firstly, Cash et al.~\cite{cash2013highly} introduce a new protocol called OXT that allows evaluation of boolean queries on encrypted data.
Faber et al.~\cite{faber2015rich} extend this data structure to support range queries but either leak additional information on the queried range or the result set contains false positives.  
In~\cite{demertzis_2016}, Demertzis et al. present several approaches for range queries.
The authors also evaluate the security and performance based on the OXT protocol.
The scheme that is most comparable to ours, Logarithmic-URC, is quite equal to the range query approach without false positives from~\cite{faber2015rich} and thus exhibits equal additional leakage.
We provide an experimental comparison in Section~\ref{sec:eval:subsec:searchEnc}.
Secondly, Pappas et al.~\cite{pappas2014blind} evaluate encrypted bloom filters using Secure Multiparty Computation.
However, in order to achieve practical efficiency they propose to split the server into two non-colluding parties. 
Our approach does not require any additional party.
Thirdly, Egorov et al.~\cite{egorov_zerodb_2016} presented ZeroDB. 
A database system that enables a client to perform equality and range searches with the help of \BPtrees.
ZeroDB is an interactive protocol requiring many rounds and thus is not usable for network sensitive cloud computing.
Fourthly, Arx was presented by Poddar et al.~\cite{poddar_arx_2016}.
The authors propose a binary tree that uses garbled circuits at any node to evaluate the traversal direction in a protected manner.
Any traversal destroys the visited garbled circuits wherefore the client has to provide new circuits in an additional round or with the next query.
This interactive reparation step seriously reduces the usefulness in a highly concurrent cloud scenario, because any query involving intersection ranges requires sequential processing.

\subsection{TEE-Based Applications}
\label{sec:reltee}

Trusted Database System (TDB) uses a trusted execution environment (TEE) to operate the entire database in a hostile environment~\cite{MVS00}.
While TDB encrypts the entire database storage and metadata (\eg tree structures to organize the data), TDB is not concerned with information leakage from the TEE.
Neither does TDB aim at hiding access patterns nor does it consider side channels attacks against the TEE.
Furthermore, since the entire DB operates in the TEE the trusted computing base is very large exposing a very large attack surface.

Haven is an approach to shield application on an untrusted system using SGX~\cite{baumann_shielding_2014}. The goal of Haven is to enable the execution of unmodified applications inside an SGX enclave. This technique could be used to isolated off-the-shelf databases with SGX, however, Haven does not consider information leakages through memory access patterns or interactions with the untrustworthy outside world. Furthermore, this approach limits the size of the database due to limited enclave size.

VC3 (short for verifiable confidential cloud computing) adapts the MapReduce computing paradigm to SGX~\cite{SCF15}. Mapper and Reducer entities are executed in individual enclaves, this means the data flow between them can leak sensitive information. 
While VC3 is tailored towards SGX they exclude information leakage from their adversary model.
In contrast, we provide the first work on SGX that specifically focuses on information leakage in the interaction of an enclave with other entities.

Data-oblivious machine learning for SGX was presented in~\cite{OSF+16}. 
Four machine learning algorithm have been adapted by the authors in order to prevent the exploitation of side channels. 
All data and code accesses that are dependent on sensitive data are transferred into data-oblivious accesses by using a library providing a set of data-oblivious primitives.
Access to external data, specifically input data, is addressed by randomizing the data \emph{and} always accessing all data, i.e., their solution has an complexity of $O(n)$, even for tree searches.
\name, in contrast, has a complexity of $O(\log n)$ in the tree size, by following the same approach as~\cite{OSF+16} we could trivially achieve data-oblivious access to the tree since the nodes of our tree are randomized as well.
However, we would lose a main feature of our construction: search time complexity of $O(\log n)$.


\section{Conclusion}
\label{sec:conc}
In this paper, we introduce \name~-- an approach to search for ranges and values over encrypted data using hardware support making it deployable as a secure index in an encrypted database.
We provide a formal security proof explicitly including side channels and an implementation on Intel SGX.
Our solution compares favorably with existing software- and hardware-based approaches.
We require few milliseconds even for complex searches on large data and scale to almost arbitrarily large indices.
We only leak the access pattern and our trusted code protected by SGX hardware is very small exposing a small attack surface.

	\bibliographystyle{IEEEtranS}
	\bibliography{paper}

\end{document}